\documentclass[12pt]{article}

\usepackage[american]{babel}
\usepackage[round]{natbib}
\usepackage{microtype}
\usepackage{graphicx}
\usepackage{xcolor}
\usepackage{subcaption}
\usepackage{booktabs}
\usepackage{array}
\usepackage{multirow}
\usepackage{adjustbox}
\usepackage{float}
\usepackage{mathtools}
\usepackage{amsmath,amssymb}
\usepackage{amsthm}
\usepackage[margin=1.25in]{geometry}
\usepackage{times}
\usepackage{hyperref}
\usepackage[capitalize,noabbrev]{cleveref}
\usepackage{common}
\usepackage{algorithm2e}
\usepackage{autonum}
\usepackage{enumitem}
\usepackage{tikz}
\usetikzlibrary{positioning, arrows.meta}
\colorlet{revisionblue}{black}
\colorlet{REVISIONBLUE}{revisionblue}

\newcommand{\rev}[1]{{#1}}
\newcommand{\revheading}[1]{\texorpdfstring{\rev{#1}}{#1}}
\newenvironment{revblock}{\begingroup}{\endgroup}
\newif\ifarxiv
\arxivtrue

\theoremstyle{plain}
\newtheorem{theorem}{Theorem}
\newtheorem{proposition}[theorem]{Proposition}

\theoremstyle{definition}
\newtheorem{definition}[theorem]{Definition}

\theoremstyle{remark}
\newtheorem{remark}[theorem]{Remark}

\title{Improving online FDR procedures \\ via online analogs of e-closure and compound e-values}
\author{
    Ziyu Xu\thanks{Department of Statistics and Data Science, Carnegie Mellon University. \texttt{xzy@cmu.edu}}
    \and
    Lasse Fischer\thanks{Competence Center for Clinical Trials Bremen, University of Bremen. \texttt{fischer1@uni-bremen.de}}
    \and
    Aaditya Ramdas\thanks{Departments of Statistics and Data Science, and Machine Learning, Carnegie Mellon University. \texttt{aramdas@cmu.edu}}
}
\date{}

\begin{document}
\renewcommand{\thefootnote}{\fnsymbol{footnote}}
\maketitle
\setcounter{footnote}{0}
\renewcommand{\thefootnote}{\arabic{footnote}}

\begin{abstract}
In many scientific applications, hypotheses are generated and tested continuously in a stream. We develop a framework for improving online multiple testing procedures with false discovery rate (FDR) control under arbitrary dependence. Our approach is two-fold: we construct methods via the online e-closure principle, as well as a novel formulation of online compound e-values that is defined through donations. This yields strict power improvements over state-of-the-art e-value and p-value procedures while retaining FDR control. We further derive algorithms that compute the decision at time $t$ in $O(\log t)$ time, and we demonstrate improved empirical performance on synthetic and real data.
 \end{abstract}
\tableofcontents

\section{Introduction}

Large-scale hypothesis testing has become prevalent in multiple industries, such as A/B testing, genomics, clinical trials, neuroscience, \rev{and online monitoring}, where the scientist wishes to filter for hypotheses where a true discovery is made and further investigation or action is merited. \rev{Online error control is particularly important in platform clinical trials \citep{robertson2023onlineplatform,zehetmayer2022online} and streaming anomaly detection \citep{lavin_evaluating_real-time_2015}, where evidence arrives sequentially and testing decisions may trigger costly follow-up actions.} Automated methods for doing so have become particularly prevalent as of late, with the rapidly improving capabilities of large language models (LLMs) to either act as an autonomous agent for generating hypotheses that it can then investigate, or as part of a human-in-the-loop system where the human scientist generates hypotheses and the model helps to triage them. In either case, hypotheses are often formulated in a sequential manner where the scientist (or LLM agent) generates some candidate hypotheses, gathers or analyzes data in the context of hypotheses, and then continues to generate more hypotheses to further elucidate their model of the world. Thus, one must provide statistical guardrails to such a system to ensure that the agent or system does not make too many false discoveries, and overfit their conclusions to the data at hand.

This motivates the problem of online multiple hypothesis testing \citep{foster_a-investing_procedure_2008}, where we assume that there is an infinite stream of hypotheses $H_1, H_2, \dots$, and a new hypothesis arrives at each time step. As each hypothesis arrives, we assume we must make the decision whether to reject or accept the null hypothesis. Let $\mathcal{N} \subset \naturals \coloneqq \{1, 2, \dots\}$ denote the set of null hypotheses, i.e., the subset of indices where the null hypothesis is true. Consequently, we output a monotonically growing sequence of discovery sets $\emptyset= R_0 \subseteq R_1 \subseteq R_2, \dots$ for each time step $t \in \naturals$, where $R_t$ contains $t$ if and only if we reject the $t$th null hypothesis. 

\paragraph{Error metrics to control.} The false positive criterion we wish to control is an online version of the false discovery rate (FDR) \citep{benjamini_controlling_false_1995}, an error criterion that has been central to statistical methodology in the offline multiple testing setting, where the number of hypotheses is known beforehand, for several decades. 
We define the FDR, along with the false discovery proportion (FDP), as follows:
\begin{align}
    \FDP_S(R) \coloneqq \frac{|S \cap R|}{|R| \vee 1}, \qquad \FDR(R) \coloneqq \expect\left[\FDP_{\Ncal}(R)\right].
\end{align}
In the above notation for FDP, $S$ is a candidate set of null hypotheses and $R$ is the discovery set the error metric is being evaluated on. The FDR is defined as the expectation of the FDP on the true null hypotheses $\Ncal$. In the online setting, we output a novel discovery set at each time step $t$, which motivates the following online error metric that was recently proposed by \citet{fischer_online_generalization_2025}:
\begin{align}
    \SupFDR(\mathbf{R})\coloneqq \expect\left[\sup_{t \in \naturals} \FDP(R_t)\right].
\end{align}

\begin{figure*}[t] 
\centering
\begin{tikzpicture}[
    node distance = 2.4cm and 1.1cm,
    box/.style = {draw, rounded corners, align=center, minimum width=4cm, minimum height=1cm, font=\footnotesize},
    title/.style = {font=\bfseries\footnotesize, align=center}
]

\node[title] (title1) {Existing offline result};
\node[title, right=of title1, xshift=0.7cm] (title2) {Our online extension};
\node[title, right=of title2, xshift=0.85cm] (title3) {Our online shortcut};

\node[box, below=0.2cm of title1] (b1) {FDR e-Closure Principle \\
\citep{xu_bringing_closure_2025a}};
\node[box, right=of b1] (b2) {Online SupFDR e-Closure \\ Principle
(Theorem~\ref{fact: online supfdr e-closure})};
\node[box, right=of b2] (b3) {Dynamic programming \\ (Theorem~\ref{theo: closed_elond_supfdr})};

\node[box, below=0.2cm of b1] (a1) {FDR control with \\ compound e-values \\ \citep{ignatiadis_asymptotic_compound_2025}};
\node[box, right=of a1] (a2) {SupFDR control with \\ $\gamma$-online compound e-values \\ (Proposition~\ref{prop: online-weighted-sc-compound})};
\node[box, right=of a2] (a3) {$\boldsymbol{\gamma}$-weighted donations \\ (Theorem~\ref{thm:donation-elond})};

\draw[->, thick] (a1) -- (a2);
\draw[->, thick] (a2) --  (a3);

\draw[->, thick] (b1) -- (b2);
\draw[->, thick] (b2) -- (b3);

\coordinate (paperSW) at ([xshift=-0.25cm,yshift=-0.5cm]a2.south west);
\coordinate (paperNE) at ([xshift=1cm,yshift=0.1cm]title3.north east);
\draw[dashed, rounded corners] (paperSW) rectangle (paperNE);
\node[anchor=south west, font=\footnotesize, fill=white, inner sep=1pt]
    at ([xshift=0.08cm,yshift=0.02cm]paperSW) {(this paper)};

\end{tikzpicture}
\caption{Summary of the paper's main technical contributions.\label{fig: contribution diagram}}
\end{figure*}
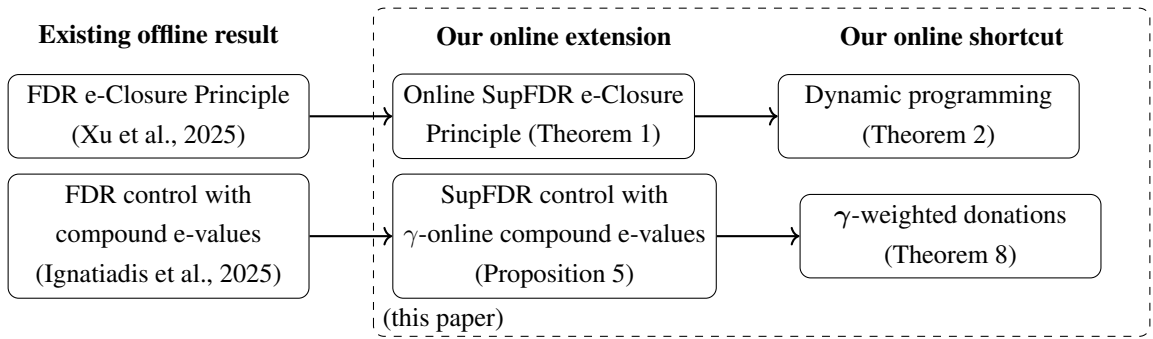

Online multiple testing with false discovery rate (FDR) control has been studied extensively in recent work \citep{xu_online_multiple_2023}. In this problem, we receive hypotheses in a stream, along with some data associated with those hypotheses. We then wish to produce a discovery set with control of the FDR at a fixed level $\delta \in (0,1]$, and maximize the number of discoveries that we make. An interesting challenge that the online setting presents, in contrast to the classical offline multiple testing setting where one possesses a fixed, known number of hypotheses beforehand, is how to define the notion of FDR, now that one will output multiple discovery sets (usually one per time step).

\citet{xu_online_multiple_2023} initiated a line of study in the use of e-values for onlineFDR control, a weaker form of FDR control for the online multiple testing setting (although \citet{fischer_online_generalization_2025} later showed that their methods controlled SupFDR as well), and assumed that the dependence between the data collected for different hypotheses was unknown. E-values are nonnegative random variables with expectation at most 1 under the null hypothesis, and have been shown to be useful for multiple testing in both offline and online settings --- see \citet{ramdas_hypothesis_testing_2025} for an overview. Practically, many settings where online multiple testing is utilized (e.g., platform clinical trials \citep{robertson2023onlineplatform, zehetmayer2022online}) involve adaptive and sequential data collection, where e-values are more natural to construct than p-values \citep{ramdas_admissible_anytime-valid_2020}. Further, offline e-closure ideas \citep{xu_bringing_closure_2025a} suggest a unified way to characterize and improve e-value procedures, which we now adapt to the online SupFDR setting.
\begin{revblock}
\paragraph{A brief overview of multiple testing.} Multiple testing methods provide error guarantees when many hypotheses are tested together, and classically it is assumed that one knows the finite number of hypotheses a priori, i.e., offline multiple testing. The benchmark method for this setting is the Benjamini-Hochberg (BH) procedure \citep{benjamini_controlling_false_1995}, which is valid under independence or positive dependence among the p-values for each hypothesis, while the Benjamini-Yekutieli (BY) correction \citep{benjamini_control_false_2001} provides an analog that is valid under arbitrary dependence among the p-values. Another natural criterion for multiple testing is family-wise error rate (FWER) control, which controls the probability of making even one false discovery. FDR methods became popular as a less conservative alternative to FWER methods, since FWER control can be overly stringent in many applications \citep{benjamini_controlling_false_1995,lehmann_generalizations_2005}. In online testing, hypotheses arrive sequentially and decisions must be made as the stream unfolds. In this setting, e-LOND is the standard e-value LOND procedure under arbitrary dependence, and r-LOND is a reshaped p-value LOND procedure that achieves arbitrary dependence validity by using reshaping functions. We recall both procedures formally below before presenting our improvements.
\end{revblock}

\paragraph{Contributions} In this paper we introduce a framework that improves the power of existing online multiple testing procedures with FDR control. We first introduce an online e-closure principle for SupFDR control and apply it to improve a wide variety of online e-value and calibrated p-value procedures. \rev{This is not a direct reuse of the offline e-closure construction: the online setting requires an increasing e-collection over an infinite stream, validity uniformly over time, and test levels that can be recomputed sequentially using only past data.} This methodology yields strict improvements over the status quo, but the resulting closed procedures are computationally expensive and may require $O(t^2)$ time to compute the rejection decision at time $t$. This quickly becomes costly in long streams.

Thus, another of the key contributions of this paper is deriving a practical algorithm that can improve power while being computationally tractable in the online setting. We provide a novel framework for generally improving e-value based online multiple testing procedures with FDR control, and show that it has practical power improvements over existing methods while requiring only $O(\log t)$ time to compute the rejection decision at time $t$. Our method is based on a novel notion of \emph{online compound e-values}, which generalizes the notion of compound e-values for offline multiple testing \citep{ignatiadis_asymptotic_compound_2025} to the online setting. We show that online compound e-values can be used to generally improve online multiple testing procedures with FDR control, and demonstrate its empirical performance on both synthetic and real data. We also provide a user-friendly implementation of our methods in the accompanying code, and demonstrate that online compound e-value based methods are also computationally more efficient in practice.

Our primary contributions are as follows, and are also visually summarized in \Cref{fig: contribution diagram}.
\begin{enumerate}
    \item \emph{Closure-based strict power improvements for standard online testing.} We introduce an online e-closure principle for SupFDR control and use it to formulate strict improvements over e-LOND \citep{xu_online_multiple_2023} and r-LOND \citep{zrnic_asynchronous_online_2021}. We derive novel explicit formulations for the next test level, which involves an optimization problem over subsets of the first $t - 1$ hypotheses. Since naive optimization would require computation that is exponential in $t$, we also provide dynamic programming decomposition that only requires $O(t^2)$ time to compute the next test level.
    \item \emph{A computationally efficient donation framework for strict improvement.} Since $O(t^2)$ computation remains inefficient for moderately large values of $t$ (i.e., in the thousands), we introduce a novel framework based on donations and online compound e-values, and show that it can also strictly improve e-LOND and r-LOND while remaining computationally tractable, requiring only $O(\log t)$ computation per time step. A central technical ingredient is our notion of \emph{online compound e-values}, which powers the donation framework.
    \item \emph{Extensions beyond the standard online multiple testing setting.} Our donation framework is not restricted to only improving standard online multiple testing algorithms. We extend our framework to variants of the online multiple testing problem such as the acceptance-to-rejection setting of \citet{fischer_online_generalization_2025} and the decision deadlines setting of \citet{fisher_online_control_2022}, where we also derive strict improvements of existing algorithms. Lastly, we show that we can construct an efficient version of eBH \citep{wang_false_discovery_2022} that is strictly more powerful for offline multiple testing, and is computationally more efficient than the $\CeBH$ procedures of \citet{xu_bringing_closure_2025a}; for $m$ hypotheses, we require $O(m\log m)$ computation as opposed to the $O(m^2 \log m)$ required by $\CeBH$.
\end{enumerate}
We then demonstrate the empirical performance of our method on both synthetic and real data, where we see that both the power and computational improvements of the methods in this paper are nontrivial. \rev{\Cref{tab:method-summary} gives a compact summary of all improved procedures and their computational costs.} By default, theorem proofs are deferred to \Cref{sec:deferred-proofs}.

\begin{table}[t]
\centering
\begin{revblock}
\footnotesize
\setlength{\tabcolsep}{1pt}
\renewcommand{\arraystretch}{0.95}
\newcolumntype{L}[1]{>{\raggedright\arraybackslash}p{#1}}
\begin{adjustbox}{max width=\columnwidth}
\begin{tabular}{@{}L{0.15\columnwidth}L{0.17\columnwidth}L{0.14\columnwidth}L{0.23\columnwidth}L{0.24\columnwidth}@{}}
\toprule
Setting & Base & Type & Method & Runtime \\
\midrule
Standard &
\ELOND{} &
Closed &
$\CeLOND$ &
$O(t^2)$ \\[1.5pt]
 &  &
Donation &
donation \ELOND{} &
$O(\log t)$ \\[15pt]
 & \rLOND{} &
Closed &
$\overline{\textnormal{r-LOND}}$ &
$O(t^2)$ \\[1.5pt]
 &  &
Donation &
donation \rLOND{} &
$O(\log t)$ \\
\midrule
ARC &
online \eBH{} &
Donation &
donation online \eBH{} &
$O(t\log t)$ \\
\midrule
Decision deadlines &
e-TOAD &
Donation &
donation e-TOAD &
$O(t\log t)$ \\
\midrule
Offline &
\eBH{} &
Closed &
$\CeBH$ &
$O(m^2 \log m)$ \\[1.5pt]
 &  &
Donation &
donation \eBH{} &
$O(m \log m)$ \\
\bottomrule
\end{tabular}
\end{adjustbox}
\caption{\rev{Method summary of improved procedures grouped by setting and base procedure. For online methods, the ``Runtime'' column gives the asymptotic complexity of the $t$th online decision; cumulative costs over $m$ online hypotheses are obtained by summing these per-decision costs. For offline methods, $m$ is the number of hypotheses in the batch.}}
\label{tab:method-summary}
\end{revblock}
\end{table}

\section{Improving e-value based procedures}

Each arriving hypothesis $H_t$ is paired with either an e-value $E_t$ or a p-value $P_t$, depending on the procedure under consideration. Formally, the $t$th e-value and p-value satisfy the following properties, respectively:
\begin{align}
    &\expect[E_t] \leq 1\text{ if }t \in \mathcal{N}. \label{eq:evalue-def}\\
    &\prob(P_t \leq s) \leq s \text{ for all }s \in [0, 1]\text{ if }t \in \mathcal{N}.
\end{align}

Online multiple testing algorithms can be viewed as producing a sequence of test levels $\boldsymbol{\alpha}$, where $\alpha_t$ is the level used at time $t$. For e-value procedures, we reject $H_t$ when $E_t \geq \alpha_t^{-1}$; for p-value procedures, we reject $H_t$ when $P_t \leq \alpha_t$. Thus, the induced discovery sets are:
\begin{align}
    R_t &\coloneqq 
    \begin{cases}
    \{i \in [t] : E_i \geq \alpha_i^{-1}\} & \text{ for e-values}\\
    \{i \in [t] : P_i \leq \alpha_i\}& \text{ for p-values}
    \end{cases},
\end{align}
where $[t] \coloneqq \{1, \ldots, t\}$. We say that a procedure with discovery sets $\mathbf{R}$ \emph{strictly improves} another procedure with discovery sets $\mathbf{R}'$ if $R_t \supseteq R_t'$ for all $t \in \naturals$ almost surely, and there exists a data-generating distribution such that $\prob(\exists t \in \naturals : R_t \supset R_t') > 0$.
Under unknown dependence between e-values, e-LOND controls SupFDR \citep{xu_online_multiple_2023,fischer_online_generalization_2025}. The e-LOND procedure uses a sequence of test levels $\boldsymbol{\alpha}$ defined by
\begin{align}
    \alpha_t &\coloneqq \delta \gamma_t (|R_{t - 1}| \rev{+ 1}), \label{eq: elond alpha}
\end{align}
where $\boldsymbol{\gamma}$ is a fixed nonnegative sequence of user-chosen constants such that $\sum_{t \in \naturals} \gamma_t \leq 1$.
Similarly, \citet{zrnic_asynchronous_online_2021,javanmard_online_rules_2018} showed that the following sequence of test levels $\boldsymbol{\alpha}$ ensures FDR control for p-values with arbitrary dependence:
\begin{align}
    \alpha_t &\coloneqq \delta \gamma_t \beta_t(|R_{t - 1}| + 1),\label{eq: rlond alpha}
\end{align} where $(\beta_t)$ is a sequence of reshaping functions \citep{blanchard_two_simple_2008}. A function $\beta: [0, \infty) \rightarrow [0,\infty)$ is a \emph{reshaping function} if $\beta$ can be written as $\beta(r) = \int_0^r x\ d\nu(x)$ where $\nu$ is any probability measure on $[0, \infty)$. A typical choice, which is the online analog of the \citet{benjamini_control_false_2001} correction for offline multiple testing under arbitrary dependence, is $\beta_t(r) = (\lfloor r \rfloor \wedge t) / \ell_t$, where $\ell_t \coloneqq \sum_{i \in [t]} i^{-1}$ is the $t$th harmonic number.

We will demonstrate how to improve on both of these methods in the online setting in this paper, among others.
\subsection{The online SupFDR e-closure principle}\label{subsec:supfdr-e-closure}

We develop an online e-closure principle that improves on existing methods in the online setting. 
Recently, \citet{xu_bringing_closure_2025a}  proposed an e-Closure Principle and used this to improve the e-BH and the BY procedure for offline FDR control. 
\citet{fischer2024online} generalized the classical Closure Principle for FWER control \citep{marcus1976closed} to the online setting by using increasing families of local tests. In the following, we extend these ideas to introduce a SupFDR e-Closure Principle for the online setting.

Let $(\Fcal_t)_{t \in \naturals}$ be a filtration where $\Fcal_t$ denotes information available by time $t$. Let $\sigma(X)$ denote the sigma-algebra formed by a set $X$. Thus, we define each element of the filtration as the sigma-algebra $\Fcal_t = \sigma(\{E_i\}_{i \in [t]})$ if one is working with e-values, and $\Fcal_t = \sigma(\{P_i\}_{i \in [t]})$ if one is working with p-values. For the purposes of this framework, we consider a more general type of online procedure that outputs a collection $\Ccal_t \subseteq 2^{[t]}$ of candidate rejection sets at each time $t$, with $\Ccal_t$ being measurable w.r.t.\ $\Fcal_t$.
An e-collection is a family $(E_S)_{S \subset \naturals,\ |S|<\infty}$, with $E_\emptyset=0$, such that each $E_S$ is an e-value for the intersection null $H_S = \cap_{i \in S} H_i$ when $S$ is nonempty. In our online setting, we require \emph{increasing} e-collections:
\[
E_S \leq E_{S \cup S'}\text{ for all }S,S' \subset \naturals\text{ s.t. }\min S' > \sup S.
\]
Given an increasing e-collection, define:
\begin{align}
    \Ccal_t \coloneqq \left\{R \subseteq [t]:  E_S \geq \frac{\FDP_S(R)}{\delta} \text{ for all }S \subseteq [t]\right\}.\label{eq:online-eclosure-col}
\end{align}
When $(E_S)$ is increasing, the collections are nested ($\Ccal_1 \subseteq \Ccal_2 \subseteq \dots$), which is needed to obtain SupFDR control.

\begin{theorem}[Online SupFDR e-closure]\label{fact: online supfdr e-closure}
    Let $(E_S)_{S \subset \naturals,\ |S|<\infty}$ be an increasing e-collection.
    Assume $E_S$ is measurable with respect to $\Fcal_{\sup(S)}$ for all finite nonempty $S$.
    Then the associated e-closure collections $(\Ccal_t)_{t \in \naturals}$ in \eqref{eq:online-eclosure-col} form an online procedure that satisfies
    \[
        \expect\left[\sup_{t \in \naturals}\sup_{R \in \Ccal_t} \FDP_{\mathcal{N}}(R)\right] \leq \delta.
    \]
    Consequently, any discovery sequence $\mathbf{R}$ with $R_t \in \Ccal_t$ for all $t \in \naturals$ controls SupFDR at level $\delta$.
\end{theorem}
Proof details are deferred to \Cref{subsec:proof-online-e-closure}.
When working with a stream of e-values, one explicit increasing e-collection is
\begin{align}
E_S = \sum_{i \in S} \gamma_{|S \cap [i]|} E_i \text{ for all finite }S\subset \naturals.\label{eq: reset closure evalue}
\end{align}
For example, $E_{\{2,3\}} = \gamma_1 E_2 + \gamma_2 E_3$.
This e-collection is increasing: if we append indices larger than $\max S$, the existing summands stay unchanged and we only add nonnegative terms.
As a result, we get the following constraint on $E_t$ for making a discovery at time $t$, i.e., for $R_{t - 1} \cup \{t\}$ to be in $\Ccal_t$:
\begin{align}
    \FDP_S(R_{t - 1} \cup \{t\}) \leq \delta E_S \text{ for all }S \in 2^{[t]}.
\end{align}
With this closure framework in place, we now begin our first construction: applying the above increasing e-collection to derive a procedure that strictly improves over e-LOND.
\begin{revblock}
To see the resulting level, the constraints with $t\notin S$ are already implied by $R_{t-1}\in \Ccal_{t-1}$. For the constraints with $t\in S$, write $S=S'\cup\{t\}$ with $S'\subseteq [t-1]$ and solve the inequality
\[
\frac{1+|S'\cap R_{t-1}|}{|R_{t-1}|+1}
\leq \delta\left(E_{S'}+\gamma_{|S'|+1}E_t\right)
\]
for $E_t$. Taking the largest required lower bound over $S'$ yields the following test levels; we call the resulting procedure \emph{$\textCeLOND$ (closed e-LOND)}.
\end{revblock}
\begin{align}
\alpha_t
&=\min_{\substack{S \subseteq [t - 1]:\\ D_t(S)>0}}
\frac{\delta \gamma_{|S| + 1} (|R_{t - 1}| + 1)}{D_t(S)},\label{eq:closed-elond-alpha}\\
D_t(S)
&\coloneqq 1 +  |S \cap R_{t - 1}| - \delta E_S (|R_{t - 1}| + 1).
\end{align}

\begin{theorem}\label{theo: closed_elond_supfdr}
    The $\textCeLOND$ procedure defined by the test levels in \eqref{eq:closed-elond-alpha} ensures SupFDR control at level $\delta$ under arbitrary dependence between e-values, and strictly improves over e-LOND when $\boldsymbol{\gamma}$ is nonincreasing.\label{thm: closed elond}
\end{theorem}
\begin{revblock}
\begin{proof}
The preceding derivation gives the exact threshold required for $R_{t-1}\cup\{t\}$ to satisfy all closure constraints involving the new index $t$. The constraints not involving $t$ are already satisfied because $R_{t-1}\in\Ccal_{t-1}$. Thus, whenever $\CeLOND$ rejects at time $t$, we have $R_t=R_{t-1}\cup\{t\}\in\Ccal_t$; if it does not reject, then $R_t=R_{t-1}\in\Ccal_t$ by monotonicity of the closure constraints. By induction, $R_t\in\Ccal_t$ for all $t$, so SupFDR control follows from \Cref{fact: online supfdr e-closure}.

It remains to compare the test levels with e-LOND. Since $R_{t-1}\in\Ccal_{t-1}$, for every $S\subseteq[t-1]$ we have $\FDP_S(R_{t-1})\leq\delta E_S$, and hence
\[
    |S\cap R_{t-1}|-\delta E_S|R_{t-1}|\leq 0 .
\]
For any nonvacuous constraint in \eqref{eq:closed-elond-alpha}, the denominator is therefore at most one, while the nonincreasingness of $\boldsymbol{\gamma}$ gives $\gamma_{|S|+1}\geq\gamma_t$. Each candidate level is consequently at least $\delta\gamma_t(|R_{t-1}|+1)$, the e-LOND level in \eqref{eq: elond alpha}. Hence $\CeLOND$ dominates e-LOND.

Strict improvement is possible already at $t=2$. On an event with positive probability where $E_1=(\delta\gamma_1)^{-1}$, both methods reject $H_1$, and the $S=\{1\}$ constraint in \eqref{eq:closed-elond-alpha} is vacuous because
\[
    1+|S\cap R_1|-\delta E_S(|R_1|+1)=2-2\delta\gamma_1E_1=0.
\]
Thus $\alpha_2^{\CeLOND}=2\delta\gamma_1>2\delta\gamma_2=\alpha_2^{\textnormal{e-LOND}}$ when $\gamma_1>\gamma_2$. A two-point construction with positive probability on this value of $E_1$, and with $E_2$ falling between the two rejection thresholds with positive probability, gives an explicit distribution on which $\CeLOND$ rejects $H_2$ while e-LOND does not.
\end{proof}
\end{revblock}

\paragraph{Computation via dynamic programming.} Although we have shown that $\CeLOND$ strictly improves over e-LOND, it is expensive to compute. While \citet{xu_bringing_closure_2025a} developed computational shortcuts for the offline e-closure principle, it is not obvious how to adapt such shortcuts to the online setting due to the weighted merging of e-values involved in solving the maximization problem in \eqref{eq:closed-elond-alpha}. A naive computation would require exponential time in $t$, while dynamic programming reduces this to a worst-case $O(t^2)$ computation per time step. To carry out this computation, we define:
\begin{align}
v_t(i, k) \coloneqq \max_{S \subseteq [i]: |S| = k} |S \cap R_{t - 1}| - \delta E_S  (|R_{t - 1}| + 1).
\end{align} The $\CeLOND$ choice of $\alpha_t$ in \eqref{eq:closed-elond-alpha} can be computed as
\begin{align}
    \alpha_t = \min_{\substack{k \in \{0\}\cup [t - 1]:\\ 1 + v_t(t - 1, k)>0}} \frac{\delta \gamma_{k + 1}(|R_{t - 1}| + 1)}{1 + v_t(t - 1, k)}.
\end{align} We can compute $v_t(i, k)$ for $i \in [t - 1]$ and $k \in \{0\}\cup [t - 1]$ using the following dynamic programming formula (for $k \geq 1$, and with $v_t(i, 0) = 0$):
\begin{align}
    &v_t(i, k) = \max\left\{v_t(i - 1, k), \right.\\
    &\left. v_t(i - 1, k - 1) + \ind\{i \in R_{t - 1}\} - \delta \gamma_{k} E_i (|R_{t - 1}| + 1)\right\}.
\end{align}
Thus, computing the dynamic program requires $O(t^2)$ time. However, this quickly becomes computationally costly in practice, motivating the need for more efficient methods.
\begin{remark}
One minor drawback of the e-collection in \eqref{eq: reset closure evalue} is that showing strict improvement in \Cref{thm: closed elond} requires $\boldsymbol{\gamma}$ to be nonincreasing. To avoid this restriction, we can define
\begin{align}
E_S = \sum_{i \in S} \gamma_{i, S} E_i,
\end{align} 
where $\gamma_{i, S} \coloneqq \sum_{j \in N_S(i)} \gamma_j$ and
$N_S(i) \coloneqq \{\max(S \cap [i - 1]) + 1, \dots, i\}$, with the convention $\max(\emptyset)=0$.
Intuitively, this is still a weighted sum of the e-values in $S$, but each weight aggregates $\gamma$-mass between consecutive selected indices in time order. This requires different computational shortcuts for test-level computation; see \Cref{sec: online gamma comp}.
\end{remark}

\subsection{Compound e-values via donation}\label{subsec:compound-evalues-donation}

So far, online multiple testing methods have primarily treated e-values either as direct inputs or as intermediaries for improving p-value procedures. However, we will instead leverage the notion of \emph{compound e-values} to substantially improve power. 
In offline multiple testing with a fixed $m \in \naturals$, one calls nonnegative random variables $(\tilde{E}_1, \dots, \tilde{E}_m)$ \emph{compound e-values} if $\sum_{i \in \mathcal{N}} \expect[\tilde{E}_i] \leq m$ \citep{ignatiadis_asymptotic_compound_2025}. This relaxes the usual e-value condition in \eqref{eq:evalue-def}. We now introduce the online, $\boldsymbol{\gamma}$-weighted analog together with donation sequences used to construct it.
Fix a nonnegative sequence $\boldsymbol{\gamma} = (\gamma_t)_{t \in \naturals}$ with $\sum_{t \in \naturals}\gamma_t \leq 1$.
\begin{definition}[$\boldsymbol{\gamma}$-online compound e-values and $\boldsymbol{\gamma}$-weighted donations]\label{def:online-compound-and-donation}
For this fixed $\boldsymbol{\gamma}$:
\begin{enumerate}[label=(\roman*)]
    \item $\tilde{\mathbf{E}} = (\tilde{E}_t)_{t \in \naturals}$ is a stream of \emph{$\boldsymbol{\gamma}$-online compound e-values} if $\sum_{i \in \mathcal{N}} \gamma_i \expect[\tilde{E}_i] \leq 1$.
    \item $\mathbf{B} = (B_t)_{t \in \naturals}$ is a \emph{$\boldsymbol{\gamma}$-weighted donation} if $\sum_{i \in [t]} \gamma_i B_i \leq 0$ for all $t \in \naturals$, and $B_t \geq -(E_t \wedge 1)$ for all $t \in \naturals$.
\end{enumerate}
\end{definition}

We first note that weighted self-consistency applied to online compound e-values has valid SupFDR control. Define the following collection of weighted self-consistent discovery sets:
\begin{align}
    \Ccal \coloneqq \left\{R \subset \naturals: \tilde{E}_t \geq \frac{1}{\delta \gamma_t |R|}\text{ for all }t \in R\right\} \label{eq:online weighted self-consistency}
\end{align}

\begin{proposition}\label{prop: online-weighted-sc-compound}
    Let $\tilde{E}_1, \tilde{E}_2, \dots$ be a stream of $\boldsymbol{\gamma}$-online compound e-values. Then, $$\expect\left[\sup_{R \in \boldsymbol{\Ccal}} \FDP(R)\right] \leq \delta.$$    
\end{proposition}
A full proof is provided in \Cref{subsec:proof-online-weighted-sc-compound}.
Now that we have shown that online compound e-values can be used to control SupFDR, we introduce a construction that, when combined with the test levels of e-LOND, strictly dominates e-LOND applied only to the original e-values. This construction is computationally efficient, requiring only $O(\log t)$ time to compute the compound e-value and hence the rejection decision at each time step. Further, it is robust to unknown dependence between e-values.
Now, we will define how to construct online compound e-values via \emph{donation}.
Let $\mathbf{B}$ be any $\boldsymbol{\gamma}$-weighted donation sequence as in \Cref{def:online-compound-and-donation}. Note that $\mathbf{B}$ can be arbitrarily dependent with $\mathbf{E}$. We first note the following property.
\begin{proposition}\label{prop: donation compound}
    Let $\tilde{E}_t = E_t + B_t$. Then, for all $t \in \naturals$, $\tilde{\mathbf{E}}$ is a valid sequence of $\boldsymbol{\gamma}$-online compound e-values. 
\end{proposition}
Proof details are deferred to \Cref{subsec:proof-donation-compound}.

As a result, we can choose any $\boldsymbol{\gamma}$-weighted donation $\mathbf{B}$ to construct online compound e-values $\tilde{\mathbf{E}}$ for e-LOND while retaining SupFDR control. Furthermore, we can take a supremum over choices of $\mathbf{B}$ and still retain FDP control. Let $\Ccal(\mathbf{B})$ be the collection of discovery sets defined by the online weighted self-consistency condition in \eqref{eq:online weighted self-consistency} with online compound e-values $\tilde{E}_t = E_t + B_t$ for a specific choice of $\mathbf{B}$. For a stream of arbitrarily dependent e-values $\mathbf{E}$, let $\mathcal{B}$ be the set of all valid $\boldsymbol{\gamma}$-weighted donations. Then we have:
\begin{proposition}\label{prop: sup donation}
    For any stream of e-values $E_1, E_2, \dots$, we have that 
    \begin{align}
        \expect\left[\sup_{\mathbf{B} \in \mathcal{B}} \sup_{R \in \Ccal(\mathbf{B})} \FDP(R)\right] \leq \delta.
    \end{align}
\end{proposition}
See \Cref{subsec:proof-sup-donation} for the proof.

As a result, we can define an algorithm that is equivalent to choosing the following test levels. Define the following ``wealth'' quantity for each $t \in \naturals$:\begin{align}
    \bar{W}_t &\coloneqq \sum_{i \in R_{t - 1}} \gamma_i\left(\left(E_i - \frac{1}{\delta \gamma_i (|R_{t - 1}| + 1)}\right) \wedge 1\right)\\
    & + \sum_{i \in [t - 1] \setminus R_{t - 1}} \gamma_i(E_i \wedge 1),
\end{align}
where $R_{t-1}=\{i\in [t-1]: E_i\geq \alpha_i^{-1}\}$ is defined by the following test levels:
\begin{align}
    \alpha_t \coloneqq \frac{\delta \gamma_t (|R_{t - 1}| + 1)}{1 - (\delta (|R_{t-1}| + 1)\bar{W}_t \wedge 1)}. \label{eq:donation-alpha}
\end{align} We refer to this procedure as \emph{donation e-LOND}. \rev{The convention $R_0=\emptyset$ makes $\bar{W}_1=0$ and hence $\alpha_1=\delta\gamma_1$. Intuitively, $\bar{W}_t$ is the largest $\gamma$-weighted amount of past evidence that can be shifted to $E_t$ while keeping all previous discoveries in $R_{t-1}$ valid and respecting the donation budget. Thus, previous e-values that exceed what is needed for their current rejection can donate excess mass, and unrejected e-values can donate up to $E_i\wedge 1$.}

\rev{Note that $(1 - (\delta (|R_{t-1}| + 1)\bar{W}_t \wedge 1))^{-1} \geq 1$, so donation e-LOND always has test levels at least as large as e-LOND. The two improvements should be viewed as complementary: neither donation e-LOND nor $\CeLOND$ dominates the other pointwise, although both dominate e-LOND.} We now have the following result.
\begin{theorem}\label{thm:donation-elond}
    Donation e-LOND controls the SupFDR at level $\delta$, and strictly improves over e-LOND.
\end{theorem}
\begin{revblock}
\begin{proof}[Proof sketch]
Suppose at time $t$ that we want the enlarged rejection set $R_{t-1}\cup\{t\}$ to be weighted self-consistent for some valid $\boldsymbol{\gamma}$-weighted donation $\mathbf{B}$. This requires
\begin{align}
    E_t+B_t
    &\geq \frac{1}{\delta\gamma_t(|R_{t-1}|+1)}, \label{eq:donation-derive-new}\\
    E_i+B_i
    &\geq \frac{1}{\delta\gamma_i(|R_{t-1}|+1)}
    \qquad\text{for each }i\in R_{t-1}. \label{eq:donation-derive-old}
\end{align}
The donation budget gives $\gamma_tB_t\leq -\sum_{i<t}\gamma_iB_i$. To make rejection of $H_t$ as easy as possible, we maximize the donation available to $t$ by choosing the smallest feasible $B_i$ for each past index. For $i\in R_{t-1}$, \eqref{eq:donation-derive-old} and $B_i\geq -(E_i\wedge 1)$ imply
\begin{align}
    B_i^{\min}
    =
    \left(\frac{1}{\delta\gamma_i(|R_{t-1}|+1)}-E_i\right)\vee (-1),
\end{align}
so
\begin{align}
    -\gamma_iB_i^{\min}
    =
    \gamma_i\wedge\left(\gamma_iE_i-\frac{1}{\delta(|R_{t-1}|+1)}\right).
\end{align}
For $i\notin R_{t-1}$, there is no rejection-preserving constraint, so $B_i^{\min}=-(E_i\wedge 1)$ and $-\gamma_iB_i^{\min}=\gamma_i(E_i\wedge 1)$. Summing these maximal past contributions yields exactly $\bar{W}_t$, and hence the largest feasible donation to the new hypothesis is $\gamma_tB_t=\bar{W}_t$, up to the point where the rejection constraint is already vacuous.
Substituting $B_t=\bar{W}_t/\gamma_t$ into \eqref{eq:donation-derive-new} gives
\begin{align}
    E_t
    &\geq
    \frac{1}{\delta\gamma_t(|R_{t-1}|+1)}-\frac{\bar{W}_t}{\gamma_t}\\
    &=
    \frac{1-\delta(|R_{t-1}|+1)\bar{W}_t}
    {\delta\gamma_t(|R_{t-1}|+1)}.
\end{align}
When $\delta(|R_{t-1}|+1)\bar{W}_t\geq 1$, the right-hand side is nonpositive, so any nonnegative e-value satisfies the constraint. This is captured by replacing the numerator with $1-(\delta(|R_{t-1}|+1)\bar{W}_t\wedge 1)$, which is equivalent to the test level in \eqref{eq:donation-alpha}. SupFDR control then follows from \Cref{prop: sup donation}; the strict improvement over e-LOND follows from the displayed test level being at least the e-LOND level, with strict inequality possible when past evidence can donate positive mass.
\end{proof}
\end{revblock}
The proof appears in \Cref{subsec:proof-donation-elond}.
\paragraph{Efficient donation computation.}
To efficiently compute the test levels in \eqref{eq:donation-alpha}, we need to update $\bar{W}_t$ efficiently at each time step. The only nontrivial component to compute is the summation over terms in $R_{t - 1}$, i.e., our term of interest is
\begin{align}
    \bar{W}^R_t &\coloneqq \sum_{i \in R_{t - 1}} \gamma_i\left(\left(E_i - \frac{1}{\delta \gamma_i (|R_{t - 1}| + 1)}\right) \wedge 1\right)\\
    &=\sum_{i \in R_{t - 1}}\gamma_i \wedge \left(\gamma_iE_i - \frac{1}{\delta(|R_{t - 1}| + 1)}\right).
\end{align}
Define 
\begin{align}
&\bar{w}_t^{(i)} \coloneqq \gamma_i \wedge \left(\gamma_iE_i - \frac{1}{\delta(|R_{t - 1}| + 1)}\right)\\
&=\begin{cases}
    \gamma_i E_i - \frac{1}{\delta(|R_{t - 1}| + 1)} & \text{ if }\gamma_i (E_i  - 1)\leq \frac{1}{\delta(|R_{t - 1}| + 1)}\\
    \gamma_i & \text{ otherwise.}
\end{cases}
\end{align}
As a result, we need to threshold on the value of $\gamma_i (E_i - 1)$ for each $i \in [t]$ to determine what value $\bar{w}_t^{(i)}$ takes on. To compute the sum of $\bar{w}_t^{(i)}$ efficiently, we maintain an augmented binary search tree where the key is $\gamma_i (E_i - 1)$ for each $i \in R_{t - 1}$. We augment each node with sums of $\gamma_i E_i$, $\gamma_i$ and a count of nodes for all nodes $i$ that are in the tree. Therefore, when split on $\gamma_i (E_i - 1)$, we have already computed our desired quantities. Consequently, the computation is simply $O(\log(|R_{t - 1}|)) \leq O(\log t)$ whenever we make a discovery, i.e., the insertion cost into the augmented binary search tree.
\section{Improving p-value based procedures}\label{sec:improving-pvalue-procedures}

Using the above results for improving e-value based procedures, we can also improve the r-LOND procedure for p-value based procedures. \citet{xu_online_multiple_2023} observed that the r-LOND procedure was equivalent to applying e-LOND where the e-values were defined by $E_t = f_{t}(P_t)$, where we let $f_{t}$ be the following calibrator similar to the calibrator for r-LOND \citet{xu_online_multiple_2023} for each $t \in \naturals$:
\begin{align}
    f_{t}(p) = \frac{\ind\{p \leq \delta \gamma_t t / \ell_t\}}{\delta \gamma_t \lceil (p \ell_t / (\delta \gamma_t)) \vee 1 \rceil}. \label{eq: by calibrator}
\end{align}
As a result, we can apply online e-closure principle to p-value based procedures and achieve SupFDR control. Notably, unlike the formulation of r-LOND as the application of e-LOND to calibrated p-values $f_t(P_t)$ for $t \in \naturals$, we instead construct a new e-collection. \rev{The calibrator index below is the rank $|S\cap [i]|$ of hypothesis $i$ within the subset $S$, rather than its global time index $i$. This rank-based choice is what makes the dynamic program in \Cref{subsec:closed-rlond-computation} depend only on the subset size, giving an $O(t^2)$ computation.} Let
\begin{align}
E_S &= \sum_{i \in S} \gamma_{|S \cap [i]|} f_{|S \cap [i]|}(P_i) \\
    &= \frac{1}{\delta}\sum_{i \in S}\frac{\ind\{P_i \leq \delta \gamma_{|S \cap [i]|} |S \cap [i]| / \ell_{|S \cap [i]|}\}}{\lceil (P_i \ell_{|S \cap [i]|} / (\delta \gamma_{|S \cap [i]|})) \vee 1 \rceil}\label{eq: closed rlond ecol}.
\end{align}
Consequently, we define \emph{$\overline{\text{r-LOND}}$ (closed r-LOND)} as
\begin{align}
\alpha_t
=\min_{\substack{S \subseteq [t - 1]:\\ \Delta E_S(t) > 0}}
    \frac{\delta \gamma_{|S| + 1}}{\ell_{|S| + 1}}
    \left\lfloor
   (\delta \Delta E_S(t))^{-1}\wedge (|S| + 1)\right\rfloor \qquad \label{eq:closed-r-lond}\\
\text{where } \Delta E_S(t) \coloneqq \frac{1 + |S \cap R_{t - 1}|}{\delta(|R_{t - 1}| + 1)} - E_S.
\end{align}
Note that the constraint set is never empty since we can always select $S = \emptyset$.
\begin{theorem}\label{thm: closed rlond}
    $\overline{\text{r-LOND}}$ controls the SupFDR at level $\delta$ for arbitrarily dependent p-values. Further, when $i \gamma_i /\ell_i$ is nonincreasing in $i$ for $i \in \naturals$, $\overline{\text{r-LOND}}$ strictly improves over r-LOND.
\end{theorem}
The proof is deferred to \Cref{subsec:proof-closed-rlond}.
We also elaborate on computational details for $\overline{\rLOND}$ in \Cref{subsec:closed-rlond-computation}, but they are similar to that of $\textCeLOND$, and it consequently requires only $O(t^2)$ computation at each time step.

Define the donation excess-wealth term $\bar{W}_t$ as we do in \eqref{eq:donation-alpha}, but using the calibrated e-values derived via $E_t = f_t(P_t)$. We then obtain the following test levels for arbitrarily dependent p-values:
\begin{align}
\alpha_t
=
\frac{\delta \gamma_t}{\ell_t}
\left(
    \left\lfloor
        \frac{|R_{t - 1}| + 1}{1 - (\delta (|R_{t-1}| + 1)\bar{W}_t \wedge 1)}
    \right\rfloor
    \wedge t
\right). \label{eq:donation-r-lond}
\end{align}
\begin{theorem}\label{thm: donation rlond}
Donation r-LOND controls the SupFDR at level $\delta$ for arbitrarily dependent p-values, and can strictly improve over r-LOND.
\end{theorem}
A proof of the above is deferred to \Cref{subsec:proof-donation-rlond}.

\section{Simulations}\label{sec: simulations}

We compare our procedures in a local dependence setup inspired by \citet{zrnic_asynchronous_online_2021}. Each hypothesis $t \in [m]$ produces a single Gaussian observation with $\mu_0 = 0$ under the null and $\mu_1 = 3$ under the alternative, where $\pi_1$ denotes the probability of the alternative being true; we set $m = 200$. For a fixed lag $L = 100$, samples within distance $L$ of one another share Gaussian copula dependence: the covariance matrix satisfies $\Sigma_{i, j} = 0.5^{|i - j|}$ for $|i - j| \leq L$ and zero otherwise, and we verify positive semidefiniteness numerically for each simulated dimension. The resulting e-values and p-values are based on Gaussian likelihood ratios and the Gaussian c.d.f. --- see \Cref{sec:local-dep-details} for details.
Each design is averaged over $n = 200$ trials. For all methods, we use the sequence where $\gamma_t = (t(t + 1))^{-1}$ for each $t \in \naturals$ unless otherwise stated. We also consider additional simulation settings in \Cref{sec: additional simulations}.

We see in \Cref{fig:local-dep-simulation} the results of applying the base procedure of e-LOND and r-LOND, as well as the power gain our procedures offer over both e-LOND and r-LOND. Both donation e-LOND and $\CeLOND$ outperform e-LOND, with $\CeLOND$'s power differential increasing over donation e-LOND as the proportion of non-nulls increases. Similarly, both donation r-LOND and $\overline{\text{r-LOND}}$ also outperform r-LOND, with $\overline{\text{r-LOND}}$'s power differential increasing over donation r-LOND as the proportion of non-nulls increases. Thus, we see that our frameworks for improving the power of both procedures have practical results.
\begin{figure*}[th]
    \centering
    \begin{subfigure}{.49\textwidth}
        \centering
        \includegraphics[width=\textwidth,trim={0 1.5cm 0 1cm},clip]{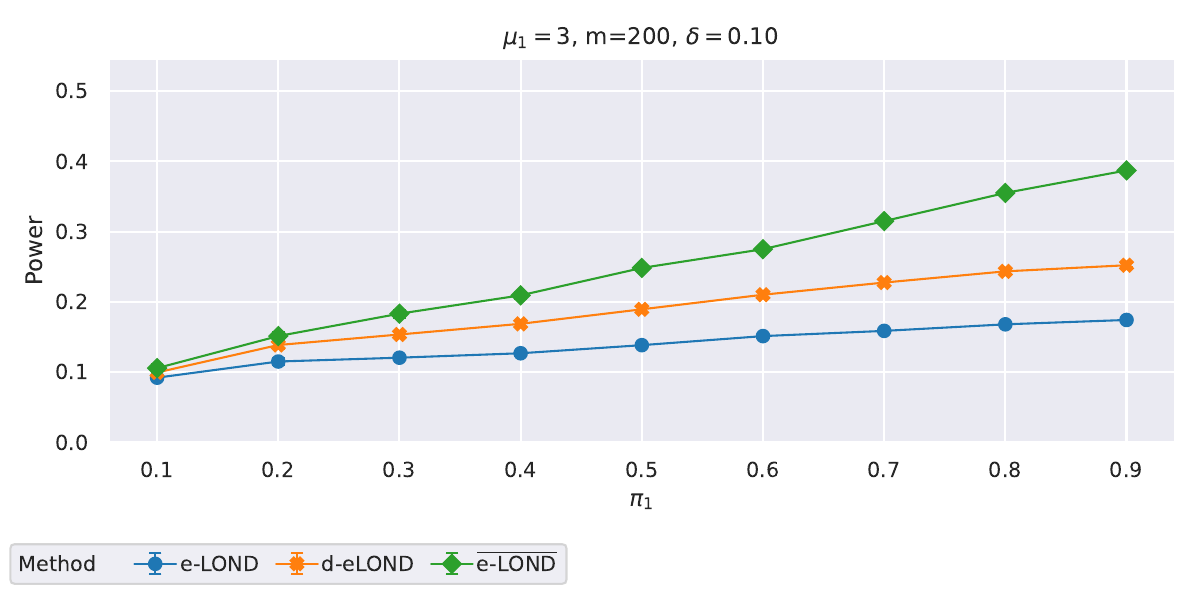}\end{subfigure}\hfill
    \begin{subfigure}{.49\textwidth}
        \centering
        \includegraphics[width=\textwidth,trim={0 1.5cm 0 1cm},clip]{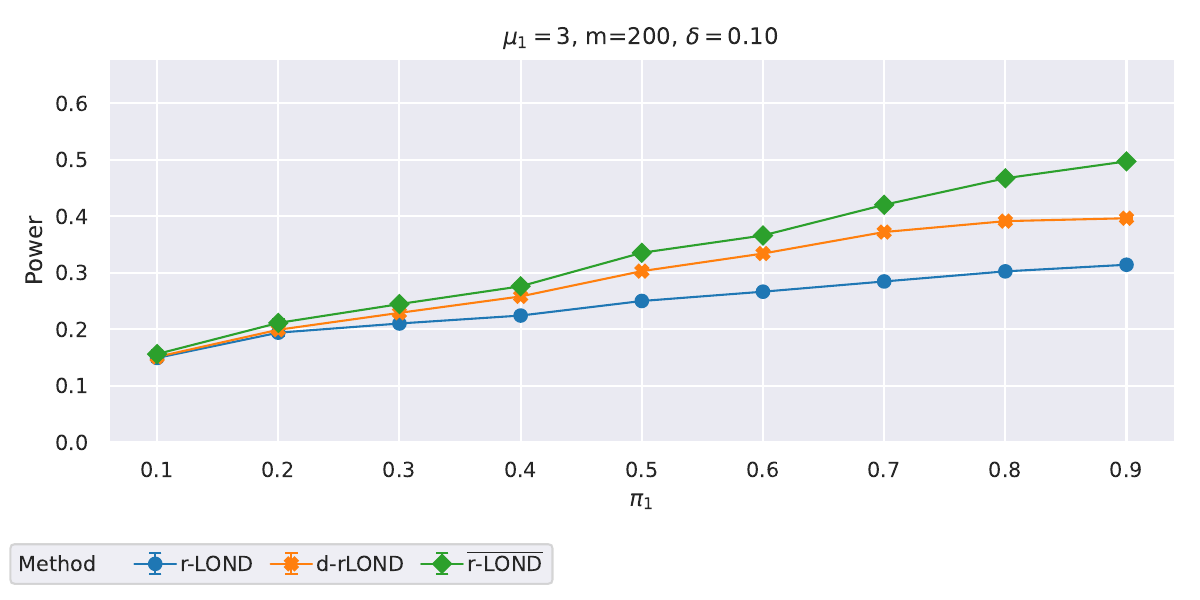}\end{subfigure}

    \begin{subfigure}{.49\textwidth}
        \centering
        \includegraphics[width=.85\textwidth,trim={0 0.5cm 504pt 1cm},clip]{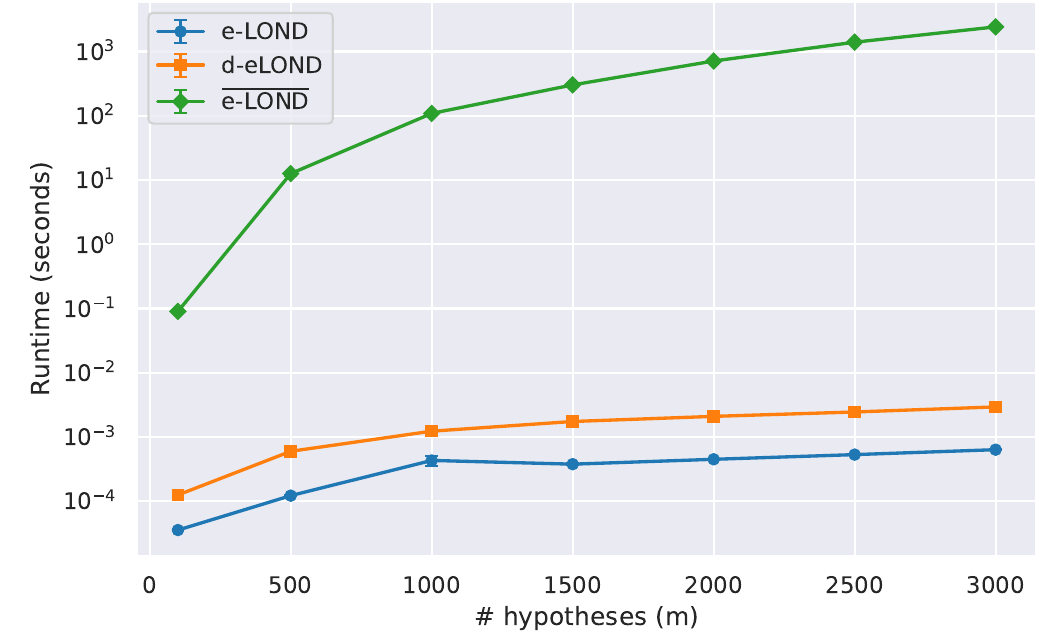}
        \caption{E-value procedures.}
    \end{subfigure}\hfill
    \begin{subfigure}{.49\textwidth}
        \centering
        \includegraphics[width=.85\textwidth,trim={504pt 0.5cm 0 1cm},clip]{figures/runtime_only_1x2.pdf}
        \caption{P-value procedures.}
    \end{subfigure}

    \caption{Local dependence simulation summary over 200 trials. The left column shows e-value procedures and the right column shows p-value procedures. The top row reports power as the non-null fraction $\pi_1$ increases (with $\mu = 3$ and $\delta = 0.1$), while the bottom row reports mean wall-clock runtime (log scale) as the number of hypotheses increases. Donation and closed variants improve power over the corresponding baselines, and donation variants remain computationally practical compared with closed variants.}
    \label{fig:local-dep-simulation}
    \label{fig:runtime-power-simulation}
\end{figure*}

\paragraph{Runtime comparison} To quantify the tradeoff between computation and power for closed vs.\ donation variants of each procedure, we run simulations using the same local dependence Gaussian data model as above, with a fixed setting of $\pi_1=0.3$. 
We choose hypothesis counts up to $m = 3000$ with $n = 80$ trials. We average the wall-clock runtime of each procedure over all trials. Wall-clock times were measured on a server with an AMD Ryzen 9 9950X CPU (16 cores, 32 threads, up to 5.756 GHz) running in a single-threaded fashion. We visualize the results of these measurements in \Cref{fig:runtime-power-simulation} --- note that the y-axis for runtime in the figure is on a log scale. We can see that as the number of hypotheses increases, the runtime of $\textCeLOND$ goes from milliseconds to on the order of an hour. On the other hand, standard e-LOND as well as donation e-LOND remain in the millisecond range. We see a similar pattern with the variants of r-LOND as well. In cumulative terms over $m$ online decisions, the closed variants scale as $O(m^3)$ when recomputed from scratch, while donation e-LOND and donation r-LOND scale as $O(m\log m)$ with the augmented-tree updates described above.

\section{\revheading{Real data experiment}}
\rev{We also show the effectiveness of our procedures on a real dataset involving online anomaly detection, following the NYC taxi demand experiment used by \citet{zhang_e-gai_e-value-based_2025} for procedure evaluation.}

\rev{The NYC taxi dataset \citep{lavin_evaluating_real-time_2015} records taxi demand over time. We use the likelihood-ratio e-value setup of \citet{zhang_e-gai_e-value-based_2025}; the dataset-specific e-value formulation and interpretation of the annotation windows are described in Appendix~\ref{subsec:real-data-annotation-metrics}.}
\ifarxiv
\begin{table}[H]
\else
\begin{table}[t]
\fi
\centering
\begin{tabular*}{\columnwidth}{@{\extracolsep{\fill}}ccc@{}}
\toprule
e-LOND & d-eLOND & $\textCeLOND$ \\
\midrule
8 & \textbf{16} & 10 \\
\bottomrule
\end{tabular*}
\caption{\rev{Discoveries coinciding with annotated anomalous periods for each e-value procedure on the NYC taxi dataset. Donation e-LOND and $\textCeLOND$ both improve over e-LOND, with donation e-LOND detecting the largest number of annotated anomalies.}}
\label{tab:real-data-summary}
\end{table}
\Cref{tab:real-data-summary} summarizes the method vs.\ the number of detected anomalies (rejections inside annotated anomaly windows) for variants of our procedures. \rev{These annotations are incomplete event windows rather than exact true-discovery labels. The largest detected anomaly count is bolded.} Both $\textCeLOND$ and donation e-LOND outperform e-LOND in terms of detected anomalies, with donation e-LOND having the largest gain. \rev{This result also illustrates that $\CeLOND$ and donation e-LOND do not dominate each other, despite each strictly improving over e-LOND.} We visualize the corresponding discoveries in \Cref{fig:real-data-paper}.

\ifarxiv
\begin{figure}[H]
\else
\begin{figure}[t]
\fi
\centering
\includegraphics[width=\columnwidth]{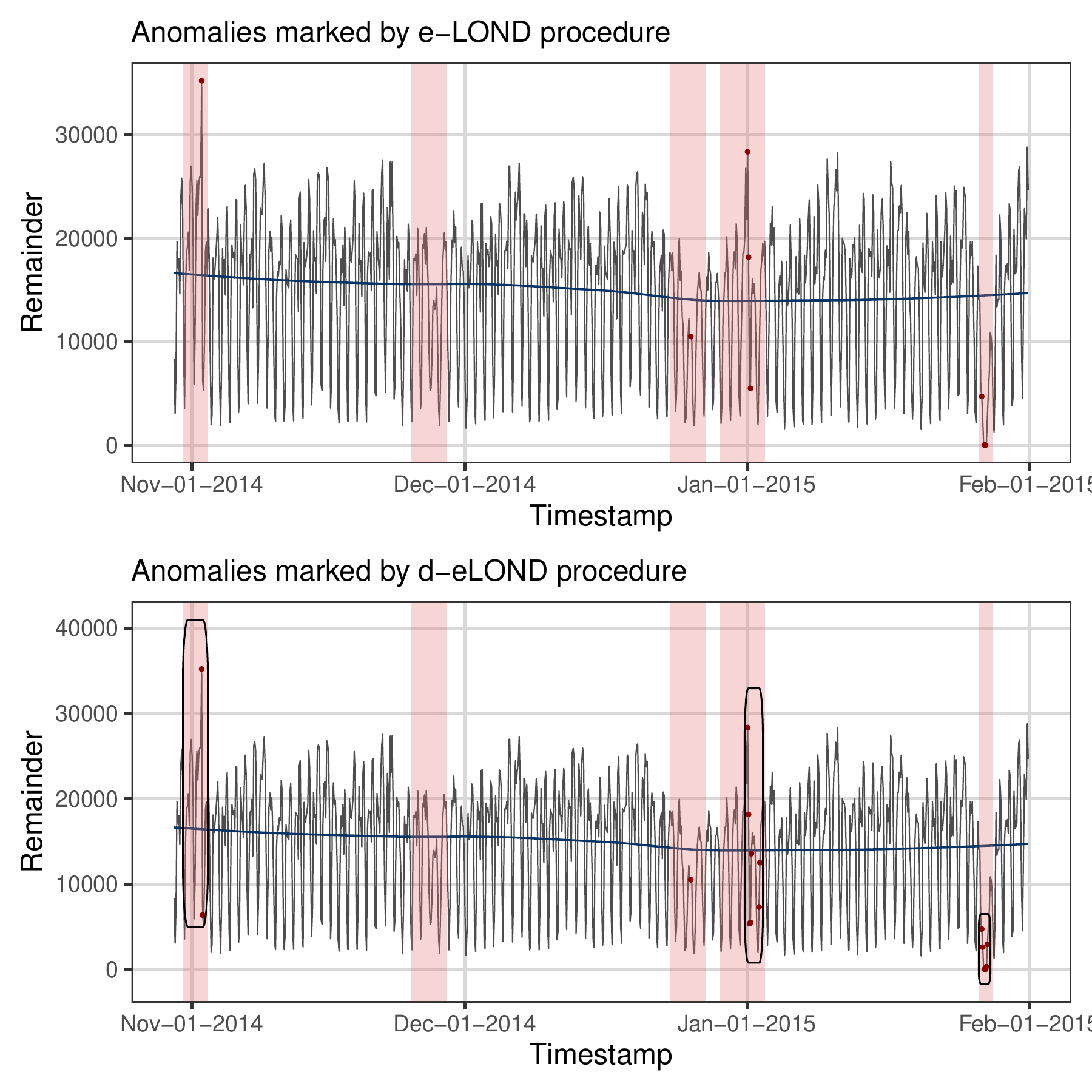}
\caption{\rev{Plot of NYC taxi usage over time, with red highlighted periods indicating anomalous periods. We can see donation e-LOND making more discoveries in these periods than e-LOND (circled in black).}}\label{fig:real-data-paper}
\end{figure}
\section{Extensions}\label{sec: extensions}

In addition to improving the aforementioned online multiple testing algorithms, the donation framework can also be used to improve the power of three other classes of algorithms.
\begin{enumerate}
    \item \textbf{Online acceptance-to-rejection (ARC) and decision deadlines.} \citet{fischer_online_generalization_2025} introduced the online ARC problem --- here, one still receives a stream of hypotheses and corresponding statistics, but no longer is forced to immediately make a decision before receiving the next hypothesis and associated statistic. 
We can use the donation to improve their online e-BH procedure in a computationally efficient manner. \citet{fisher_online_control_2022} considers an intermediate setting where each hypothesis $t$ has a (rejection) decision deadline $d_t \geq t$, and we show that we can improve the power of their procedure using donations as well. We elaborate on this in \Cref{sec: arc donation}.
    \item \textbf{Offline multiple testing.} While \citet{xu_bringing_closure_2025a} used the e-closure principle to improve the $\CeBH$ procedure for offline multiple testing, it can require quadratic time to compute the rejection set for $m$ hypotheses. Using donations, we can construct compound e-values that can be used with eBH to improve its power while computing a strictly more powerful discovery set in $O(m \log m)$ time. We elaborate on this in \Cref{sec: offline donation}.
    \item \textbf{Randomization}: \citet{xu_more_powerful_2023} introduced the notion of randomization for improving the power of multiple testing procedures, and \citet{xu_online_multiple_2023} applied it to improve both e-LOND and r-LOND. Similarly, we can apply randomization to improve the power of donation procedures. We elaborate on this in \Cref{sec: randomized donation}.
\end{enumerate}

\section{Related work}

In addition to the aforementioned works, online multiple testing with FDR control has been studied extensively in recent years. \citet{javanmard_online_rules_2018} introduced an online notion of FDR control and formulated the LORD procedure for control under independence. As a result, there has been a line of literature that has developed increasingly powerful online multiple testing algorithms under either independence or conditional validity assumption \cite{ramdas_online_control_2017a, zrnic_power_batching_2020}. \citet{ramdas_saffron_adaptive_2018} and \citet{tian_addis_adaptive_2019} formulate adaptive and discarding versions of LORD in the style of \cite{storey_direct_approach_2002} and \cite{zhao_multiple_testing_2019}, respectively, that allow the procedure to adapt to the frequency of the true null hypotheses. \rev{These LORD-type and SAFFRON-type methods are powerful under their assumptions, but they are not valid in the arbitrary dependence SupFDR setting considered in this paper, so they are not direct baselines for our main results.} These works have restricted their focus to p-value based methods. More recently \cite{zhang_e-gai_e-value-based_2025} proposed a new method under the conditional validity assumption that applies to both p-values and e-values.

A line of work has also considered designing methods under other kinds of dependency structures. In addition to defining the r-LOND procedure and considering general arbitrary dependence, \citet{zrnic_asynchronous_online_2021} considers p-value based methods for global positive dependence among test statistics and local forms of dependence. \cite{fisher_online_false_2024} show that under a different definition of positive dependence, LORD and SAFFRON also have valid FDR control. \cite{jankovic_asymptotic_online_2025} considers online FWER control for weakly dependent data, although their control is asymptotic.

Outside of \citet{xu_online_multiple_2023}, e-values have also been utilized in online multiple testing with family-wise error rate (FWER) control. \cite{fischer_admissible_online_2025} showed that e-values are necessary for constructing admissible online methods with control of FDP tail probabilities.

\paragraph{Different online multiple testing analogs of FDR.} SupFDR is a relatively new error criterion, but it is valuable in the sense that it is stronger than both of the more classical notions of FDR control that have developed for online multiple testing. The first of these is onlineFDR, which was introduced in \citet{javanmard_online_rules_2018} and is defined as the following:
\begin{align}
    \text{onlineFDR}(\mathbf{R}) \coloneqq \sup_{t \in \naturals} \FDR(R_t).
\end{align} \citet{fischer_online_generalization_2025} also considered a stopping time version of onlineFDR, which they called StopFDR and defined as:
\begin{align}
    \text{StopFDR}(\mathbf{R}) \coloneqq \sup_{\tau \in \mathcal{T}} \FDR(R_\tau),
\end{align} where $\mathcal{T}$ is the set of all stopping times with respect to the filtration generated by the data. Both onlineFDR and StopFDR control are implied by SupFDR control, as observed by \citet{fischer_online_generalization_2025}. One can see that
since $\sup_t \expect[\FDP(R_t)] \leq \expect[\sup_t \FDP(R_t)]$, and $\FDP(R_\tau) \leq \sup_{t} \FDP(R_t)$ for any stopping time $\tau$. Thus, our development of procedures for SupFDR control implies validity under previous error metrics considered in prior literature.

\section{Conclusion}
We introduced a general framework for improving several online multiple testing procedures under unknown dependence, with a focus on SupFDR control. We presented two approaches that trade off power and computational cost while both strictly improving over existing methods. Our first approach uses the online e-closure principle to produce \emph{closed} procedures that dominate their baselines, including $\CeLOND$ for e-values and $\overline{\rLOND}$ for p-values. These closed procedures can yield large power gains but are computationally expensive in long streams. Our second approach is the donation framework, which constructs online compound e-values via $\boldsymbol{\gamma}$-weighted donations and yields computationally efficient, strictly improved procedures such as donation e-LOND and donation r-LOND. Together, these results provide a principled menu of improvements: closed methods provide a strong but more expensive benchmark, while donation methods retain power improvements with efficient test-level computation and can be preferable on some instances. An interesting direction for future work is to further narrow the computational gap between closed and donation procedures, or show that there is an irreducible difference between them.
 
\bibliography{ref}

\clearpage
\appendix
\section{Methodological details of online closed procedures}
This section collects methodological details that are deferred from the main text. We first record a general weighted e-collection construction and its relation to weighted self-consistency. We then present an alternative e-collection for $\CeLOND$ (referenced in \Cref{sec: online gamma comp}) and computation details for $\overline{\rLOND}$.

\subsection{General weighted e-collections and closure enlargement}\label{subsec:weighted-closure-general}
Fix a nonnegative sequence $\boldsymbol{\gamma} = (\gamma_t)_{t \in \naturals}$ with $\sum_{t \in \naturals}\gamma_t \leq 1$.
For each finite $S \subset \naturals$, let $(\gamma_i^S)_{i \in S}$ satisfy:
\begin{enumerate}[label=(\roman*)]
    \item $\sum_{i \in S}\gamma_i^S \leq 1$ and $\gamma_i^S \geq \gamma_i$ for all $i \in S$;
    \item if $i \in S \cap T$ and $S \cap [i] = T \cap [i]$, then $\gamma_i^S = \gamma_i^T$.
\end{enumerate}
Define
\begin{align}
    E_S \coloneqq \sum_{i \in S}\gamma_i^S E_i.
\end{align}
Under $H_S$, this gives $\expect[E_S] \leq \sum_{i \in S}\gamma_i^S \leq 1$, so $(E_S)_{S \in 2^\naturals}$ is an e-collection.
Condition (ii) ensures this e-collection is increasing.

\begin{proposition}\label{prop:weighted-sc-in-closure}
    Let $(E_S)$ be defined above and let $R \subset [t]$ be nonempty.
    If $R$ is weighted self-consistent for $(\gamma_i, E_i)$, i.e.,
    \begin{align}
        \gamma_i E_i \geq \frac{1}{\delta |R|}\qquad \text{for all }i \in R,
    \end{align}
    then $R \in \Ccal_t$.
\end{proposition}

\begin{proof}
    For any $S \subseteq [t]$,
    \begin{align}
        E_S
        = \sum_{i \in S}\gamma_i^S E_i
        \geq \sum_{i \in S \cap R}\gamma_i E_i
        \geq \frac{|S \cap R|}{\delta |R|}
        = \frac{\FDP_S(R)}{\delta}.
    \end{align}
    Hence $R$ satisfies all constraints in \eqref{eq:online-eclosure-col}, so $R \in \Ccal_t$.
\end{proof}

The closure set can be strictly larger than the weighted self-consistent family.
For example, let $\gamma_1 = \gamma_2 = 1/2$,
\[
\gamma_1^{\{1\}}=\tfrac{1}{2},\quad
\gamma_1^{\{1,2\}}=\gamma_2^{\{1,2\}}=\tfrac{1}{2},\quad
\gamma_2^{\{2\}}=1,
\]
and choose $E_1 = 1/(2\delta)$, $E_2 = 3/(2\delta)$.
Then no nonempty weighted self-consistent rejection set exists, but $\{2\} \in \Ccal_2$.

\subsection{Alternative e-collection for $\CeLOND$}\label{sec: online gamma comp}
In \Cref{thm: closed elond}, we used $E_S = \sum_{i \in S} \gamma_{|S \cap [i]|}E_i$ as the e-collection to construct $\CeLOND$. That choice requires a nonincreasing $\boldsymbol{\gamma}$ sequence in the strict-improvement argument. Here we record an alternative e-collection that avoids that monotonicity requirement.

For a finite set $S = \{s_1 < \dots < s_m\}$, define $s_0 \coloneqq 0$ and
\begin{align}
    \gamma_{s_j, S} \coloneqq \sum_{\ell = s_{j - 1} + 1}^{s_j}\gamma_\ell,\qquad
    E_S \coloneqq \sum_{j = 1}^m \gamma_{s_j, S} E_{s_j}.
\end{align}
For each $t$ and $S \subseteq [t - 1]$, define
\begin{align}
    \Gamma_t(S) \coloneqq \sum_{\ell = \max(S \cup \{0\}) + 1}^{t}\gamma_\ell.
\end{align}
Then $E_{S \cup \{t\}} = E_S + \Gamma_t(S)E_t$.

\begin{theorem}\label{thm:online-gamma-comp}
    The procedure with test levels
    \begin{align}
        \alpha_t
        =\min_{\substack{S \subseteq [t - 1]:\\
        1 + |S \cap R_{t - 1}| - \delta E_S (|R_{t - 1}| + 1) > 0}}
        \frac{\delta \Gamma_t(S)(|R_{t - 1}| + 1)}
        {1 + |S \cap R_{t - 1}| - \delta E_S (|R_{t - 1}| + 1)}
    \end{align}
    controls SupFDR at level $\delta$ under arbitrary dependence, and dominates e-LOND for any nonnegative $\boldsymbol{\gamma}$ with $\sum_{i \in \naturals}\gamma_i \leq 1$. A sufficient condition for strict improvement at time $t$ is that $S = \emptyset$ attains the minimum and $\sum_{i = 1}^{t}\gamma_i > \gamma_t$.
\end{theorem}

\begin{proof}
    For each finite $S = \{s_1 < \dots < s_m\}$,
    \begin{align}
        \expect[E_S]
        \leq \sum_{j = 1}^m \gamma_{s_j, S}
        = \sum_{\ell = 1}^{s_m}\gamma_\ell
        \leq 1
    \end{align}
    under $H_S$, so $(E_S)_{S \in 2^\naturals}$ is a valid e-collection.

    For $S \subseteq [t - 1]$, requiring $R_{t - 1} \cup \{t\} \in \Ccal_t$ is equivalent to
    \begin{align}
        \FDP_{S \cup \{t\}}(R_{t - 1} \cup \{t\})
        \leq \delta(E_S + \Gamma_t(S)E_t),
    \end{align}
    i.e.,
    \begin{align}
        E_t
        \geq
        \frac{\FDP_{S \cup \{t\}}(R_{t - 1} \cup \{t\}) - \delta E_S}
        {\delta \Gamma_t(S)}.
    \end{align}
    Taking the maximum over $S \subseteq [t - 1]$ yields exactly the displayed test level, so $R_t \in \Ccal_t$ for all $t$. SupFDR control then follows from \Cref{fact: online supfdr e-closure}.

    To compare with e-LOND, use $R_{t - 1} \in \Ccal_{t - 1}$ to get
    \begin{align}
        |S \cap R_{t - 1}| \leq \delta E_S |R_{t - 1}|,
    \end{align}
    hence
    \begin{align}
        1 + |S \cap R_{t - 1}| - \delta E_S (|R_{t - 1}| + 1) \leq 1.
    \end{align}
    Also $\Gamma_t(S) \geq \gamma_t$ for every $S \subseteq [t - 1]$. Therefore each candidate term in the minimum is at least $\delta\gamma_t(|R_{t - 1}| + 1)$, which is the e-LOND level, so dominance holds.

    If $S = \emptyset$ is a minimizer and $\sum_{i = 1}^{t}\gamma_i > \gamma_t$, then
    \begin{align}
        \alpha_t
        = \delta\Big(\sum_{i = 1}^{t}\gamma_i\Big)(|R_{t - 1}| + 1)
        >
        \delta\gamma_t(|R_{t - 1}| + 1),
    \end{align}
    giving strict improvement at time $t$.
\end{proof}

\paragraph{Simulation results.} We compare e-LOND, $\overline{\text{e-LOND}}$, and the alternative-$\gamma$
closed e-LOND using the same local dependence simulation setup as in the
main simulations section in \Cref{fig:supp-alt-gamma-sim-power} and plot empirical error diagnostics in \Cref{fig:supp-alt-gamma-sim-fdr}. We see that the alternative-$\boldsymbol{\gamma}$ $\textCeLOND$ has worse power than the original $\boldsymbol{\gamma}$, while still being much more powerful than the original e-LOND procedure. All methods remain below the target $\delta = 0.1$ level in these diagnostics.
\begin{figure}[t]
    \centering
    \includegraphics[width=.6\textwidth]{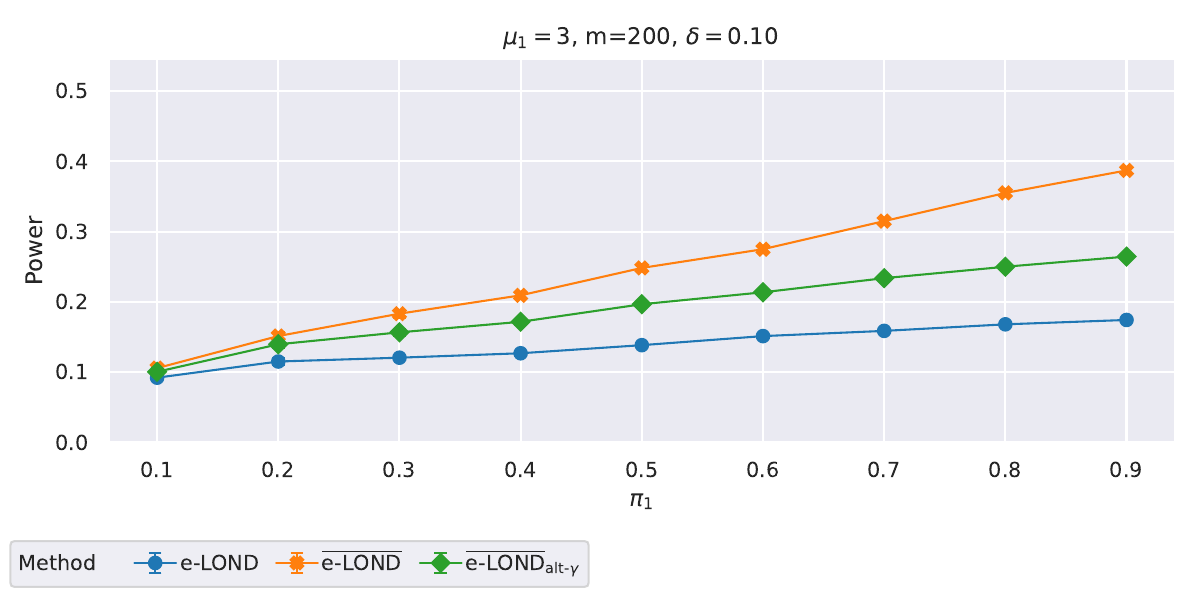}

    \includegraphics[width=.6\textwidth]{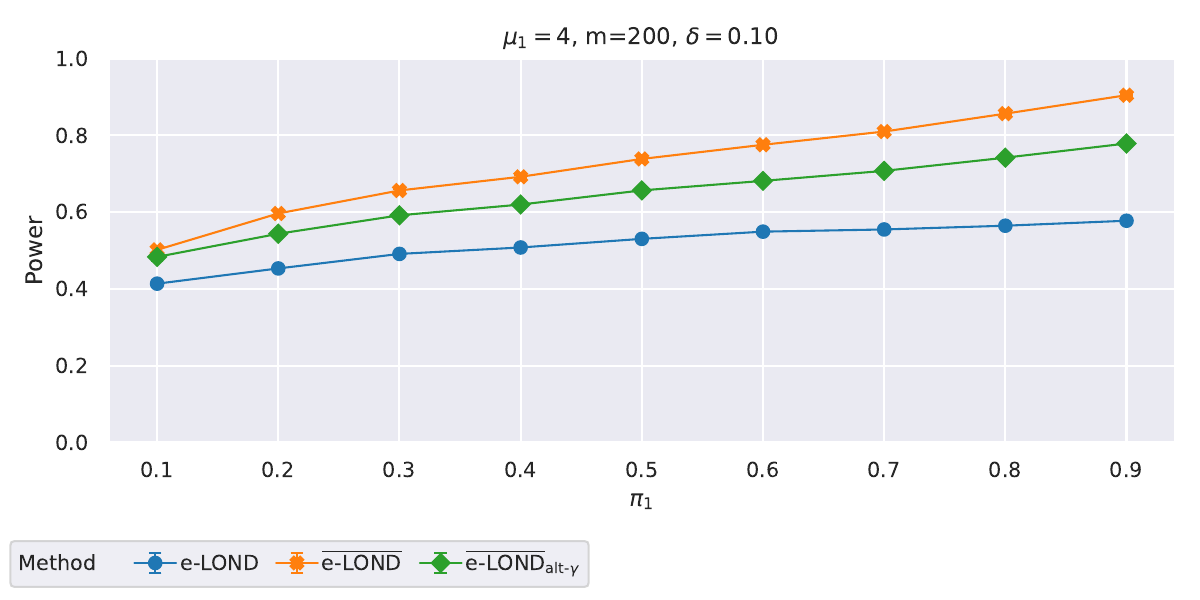}
    \caption{Power comparison for alternative choice of $\gamma$ for $\CeLOND$.}
    \label{fig:supp-alt-gamma-sim-power}
\end{figure}
\begin{figure}[t]
    \centering
    \includegraphics[width=.6\textwidth]{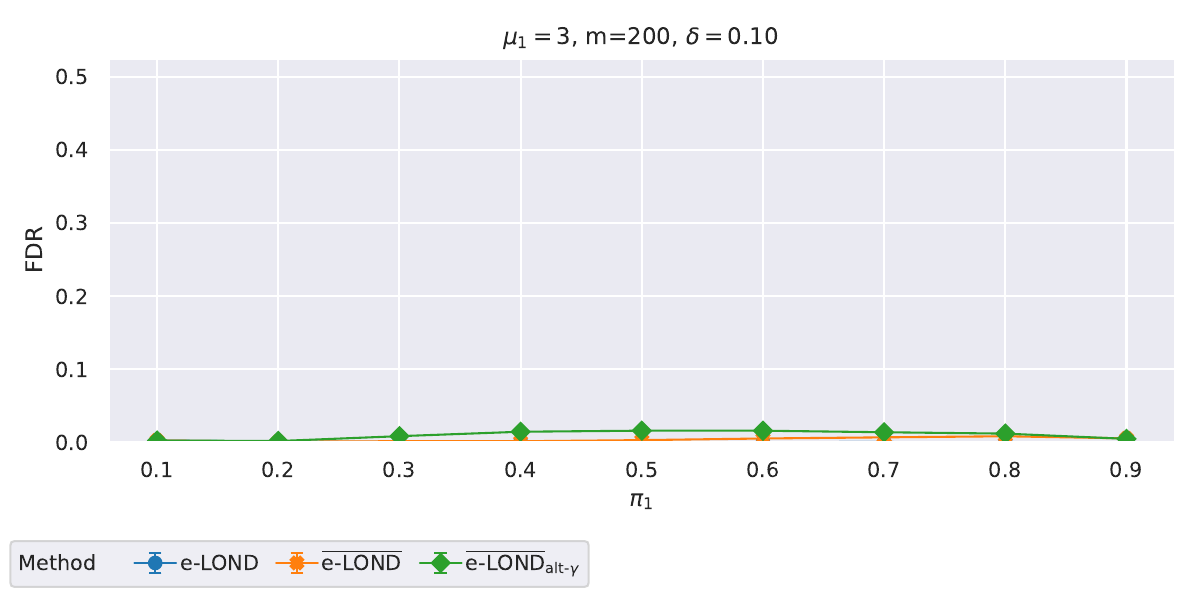}

    \includegraphics[width=.6\textwidth]{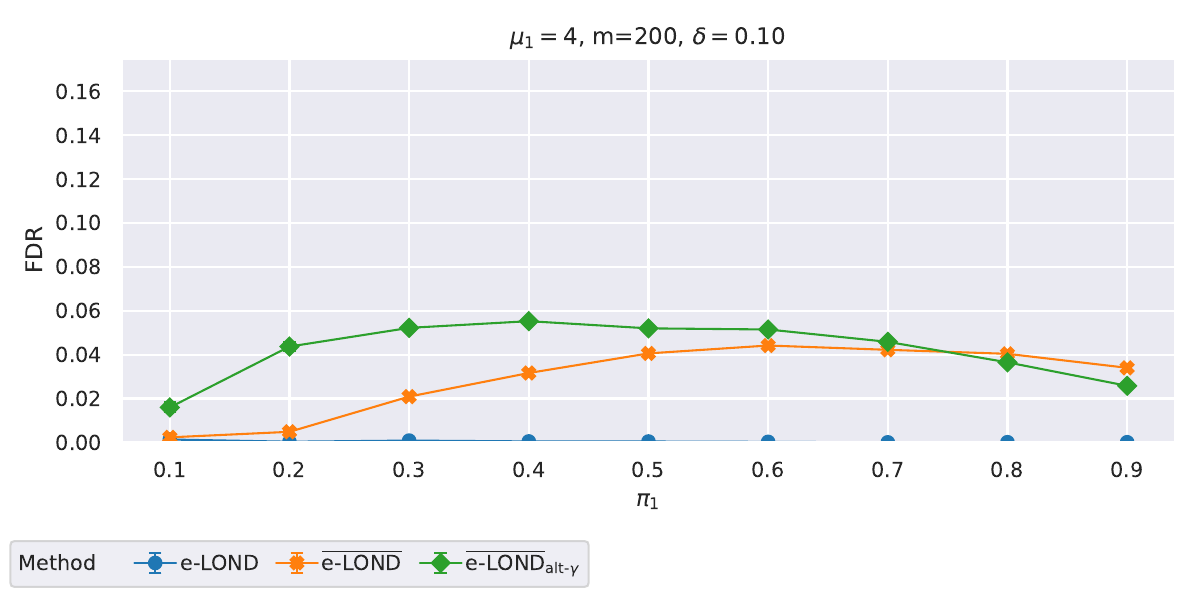}
    \caption{Empirical error diagnostics for the alternative-$\gamma$ closure comparison. All methods stay controlled at the target level $\delta = 0.1$.}
    \label{fig:supp-alt-gamma-sim-fdr}
\end{figure}

\subsection{Computation details of $\overline{\text{r-LOND}}$}\label{subsec:closed-rlond-computation}
The minimization in \eqref{eq:closed-r-lond} can be carried out in $O(t^2)$ time via a dynamic program mirroring the $\CeLOND$ case in \Cref{subsec:compound-evalues-donation}. Define for $i \in [t - 1]$ and $k \in \{0\}\cup [t - 1]$,
\begin{align}
    g_t(i, k) \coloneqq \max_{S \subseteq [i]: |S| = k} |S \cap R_{t - 1}| - \delta (|R_{t - 1}| + 1)E_S.
\end{align}
Initialize $g_t(0, 0) = 0$ and $g_t(0, k) = -\infty$ for $k > 0$. For $i \geq 1$, set $g_t(i, 0) = 0$, and for $k \in \{1, \dots, i\}$ update
\begin{align}
    g_t(i, k) = \max\Big\{ &g_t(i - 1, k),\\
        &g_t(i - 1, k - 1)
            + \ind\{i \in R_{t - 1}\}
            - (|R_{t - 1}| + 1)
                \left\lceil \left(\frac{P_i \ell_k}{\delta \gamma_k}\right) \vee 1 \right\rceil^{-1} \Big\}.
\end{align}
Then
\begin{align}
    \alpha_t^{\overline{\rLOND}}
    =
    \min\Bigg\{1,\;
    \min_{\substack{k \in \{0, \dots, t - 1\}\\ 1 + g_t(t - 1, k) > 0}}
    \frac{\delta \gamma_{k + 1}}{\ell_{k + 1}}
    \left\lfloor
    \left(
    \frac{|R_{t - 1}| + 1}{
        1 + g_t(t - 1, k)
    }\right)\wedge (k + 1)
    \right\rfloor
    \Bigg\}.\notag
\end{align}
The state space only tracks $k = |S|$, so the algorithm requires $O(t^2)$ time and $O(t)$ memory per step; in practice we restrict $i$ to indices in $R_{t - 1}$ to reduce the constant factors.

\section{Deferred proofs}\label{sec:deferred-proofs}
This section contains deferred proofs for the main theoretical results in the paper, including the closure-based improvements, the donation framework lemmas, and the p-value analogs.

\subsection{Proof of Theorem~\ref{fact: online supfdr e-closure}}\label{subsec:proof-online-e-closure}
\begin{proof}
    Since each $E_S$ is measurable with respect to $\Fcal_{\sup(S)}$, the collection $(\Ccal_t)_{t \in \naturals}$ is an online procedure.
    Fix a finite horizon $T$. For any $t\leq T$ and $R\in\Ccal_t$, let $S_t\coloneqq \Ncal\cap[t]$. Then $R\subseteq[t]$, so
    \begin{align}
        \FDP_{\Ncal}(R)=\FDP_{S_t}(R)\leq \delta E_{S_t}.
    \end{align}
    Since $S_t\subseteq \Ncal\cap[T]$ and the e-collection is increasing, $E_{S_t}\leq E_{\Ncal\cap[T]}$. Therefore
    \begin{align}
        \sup_{t\leq T}\sup_{R\in\Ccal_t}\FDP_{\Ncal}(R)
        \leq \delta E_{\Ncal\cap[T]}.
    \end{align}
    Taking expectations gives the finite-horizon bound because $E_{\Ncal\cap[T]}$ is an e-value for the finite intersection null $H_{\Ncal\cap[T]}$:
    \[
        \expect\left[\sup_{t\leq T}\sup_{R\in\Ccal_t}\FDP_{\Ncal}(R)\right]\leq\delta.
    \]
    Letting $T\to\infty$ and applying monotone convergence yields the stated SupFDR bound.
    This proves the first claim; the stated consequence for any discovery sequence with $R_t \in \Ccal_t$ is immediate.

\end{proof}

\subsection{Proof of Proposition~\ref{prop: online-weighted-sc-compound}}\label{subsec:proof-online-weighted-sc-compound}
\begin{proof}
    For a finite horizon $T$, let
    \[
        \Ccal_T\coloneqq\{R\subseteq[T]: \tilde E_i\geq(\delta\gamma_i |R|)^{-1}\text{ for all }i\in R\}.
    \]
    For every $R\in\Ccal_T$,
    \begin{align}
        \FDP_{\Ncal}(R)
        =\sum_{i\in\Ncal\cap R}\frac{1}{|R|\vee1}
        \leq \delta\sum_{i\in\Ncal\cap[T]}\gamma_i\tilde E_i.
    \end{align}
    Taking the supremum over $R\in\Ccal_T$ and then expectations yields
    \[
        \expect\left[\sup_{R\in\Ccal_T}\FDP_{\Ncal}(R)\right]
        \leq \delta\sum_{i\in\Ncal\cap[T]}\gamma_i\expect[\tilde E_i]
        \leq \delta.
    \]
    The collections $\Ccal_T$ increase to $\Ccal$, so monotone convergence gives the result.
\end{proof}

\subsection{Proof of Proposition~\ref{prop: donation compound}}\label{subsec:proof-donation-compound}
\begin{proof}
    For all $t \in \naturals$, we note that
    \begin{align}
        &\sum_{i \in \mathcal{N} \cap [t]} \gamma_i \expect[\tilde{E}_i]
        =
        \sum_{i \in \mathcal{N} \cap [t]} \gamma_i \expect[E_i + B_i]\\
        &\leq
        \sum_{i \in \mathcal{N} \cap [t]} \gamma_i \expect[E_i] - \sum_{i \in [t] \setminus \mathcal{N}} \gamma_i \expect[B_i]\\
        &\leq
       \sum_{i \in \mathcal{N} \cap [t]} \gamma_i \expect[E_i] + \sum_{i \in [t] \setminus \mathcal{N}} \gamma_i
       \leq 1.
    \end{align}
    \rev{The first inequality is due to the definition of $\mathbf{B}$ as a $\boldsymbol{\gamma}$-weighted donation. The last inequality is due to the fact that $E_i$ is an e-value for all $i \in \mathcal{N}$ and $\boldsymbol{\gamma}$ sum to at most 1.}
\end{proof}

\subsection{Proof of Proposition~\ref{prop: sup donation}}\label{subsec:proof-sup-donation}
\begin{proof}
    For any valid donation sequence $\mathbf{B}$ and any $R\in\Ccal(\mathbf{B})$, weighted self-consistency gives the pathwise bound
    \begin{align}
        \FDP_{\Ncal}(R)
        \leq
        \delta\sum_{i\in\Ncal}\gamma_i(E_i+B_i).
    \end{align}
    Taking the supremum over valid $\mathbf{B}$, the largest possible transfer into null coordinates is bounded by the amount withdrawn from nonnull coordinates, and each nonnull coordinate can withdraw at most $E_i\wedge1$. Hence
    \[
        \sup_{\mathbf{B}\in\mathcal B}\sum_{i\in\Ncal}\gamma_i(E_i+B_i)
        \leq
        \sum_{i\in\Ncal}\gamma_iE_i
        +\sum_{i\notin\Ncal}\gamma_i(E_i\wedge1).
    \]
    Taking expectations gives at most $\sum_i\gamma_i\leq1$, proving the claim.
\end{proof}

\subsection{Proof of Theorem~\ref{thm:donation-elond}}\label{subsec:proof-donation-elond}
\begin{proof}
    \rev{The proof sketch after \Cref{thm:donation-elond} derives the test level \eqref{eq:donation-alpha} by maximizing the feasible donation to the new hypothesis while keeping all previous discoveries weighted self-consistent. Equivalently, for each $t \in \naturals$, define the $\boldsymbol{\gamma}$-weighted donation sequence $\mathbf{B}^{(t)}$ by}
    \begin{align}
        B_i^{(t)} =
        \begin{cases}
        \bar{W}_t / \gamma_t & \text{ if } i = t,\\
        ((\delta\gamma_i(|R_{t - 1}| + 1))^{-1} - E_i) \vee -1 & \text{ if }i \in R_{t - 1},\\
        - (E_i \wedge 1)&\text{ if }i<t,\ i\notin R_{t - 1},\\
        0&\text{ if }i>t.
        \end{cases}
    \end{align}
    \rev{The derivation shows that $R_t\in\Ccal(\mathbf{B}^{(t)})$ exactly when $E_t\geq \alpha_t^{-1}$, with $\alpha_t$ as in \eqref{eq:donation-alpha}. Moreover, $\gamma_t B_t^{(t)} = \bar{W}_t = - \sum_{i \in [t - 1]}\gamma_i B_i^{(t)}$, so $\mathbf{B}^{(t)}$ is a valid $\boldsymbol{\gamma}$-weighted donation. Therefore, by \Cref{prop: sup donation}, donation e-LOND controls the SupFDR at level $\delta$.}

    Lastly, we show that donation e-LOND strictly improves over e-LOND. Let $R^{\textnormal{d-eLOND}}_t$ and $R^{\textnormal{e-LOND}}_t$ and $\alpha_t^{\textnormal{d-eLOND}}$ and $\alpha_t^{\textnormal{e-LOND}}$ be the discovery sets and test levels of donation e-LOND and e-LOND respectively. Then, if we assume $R_{t - 1}^{\textnormal{d-eLOND}} \supseteq R_{t - 1}^{\textnormal{e-LOND}}$ by induction, we have that
    \begin{align}
        \alpha_t^{\textnormal{d-eLOND}}
        &\rev{=}
        \frac{\delta \gamma_t (|R_{t - 1}^{\textnormal{d-eLOND}}|+ 1)}{1 - (\delta (|R_{t - 1}^{\textnormal{d-eLOND}}| + 1)\bar{W}_t \wedge 1)}\\
        &\geq
        \delta \gamma_t (|R_{t - 1}^{\textnormal{e-LOND}}| + 1)\\
        &=
        \alpha_t^{\textnormal{e-LOND}},
    \end{align}
    since $(1 - (\delta (|R_{t - 1}^{\textnormal{d-eLOND}}| + 1)\bar{W}_t \wedge 1))^{-1} \geq 1$.

    For a concrete instance where donation e-LOND make strictly more discoveries than e-LOND with positive probability, take $t = 2$ and let $\gamma_1 > \gamma_2$. If $E_1 \geq (\delta\gamma_1)^{-1}$, then both methods reject $H_1$ and $R_1 = \{1\}$, so $\alpha_2^{\textnormal{d-eLOND}} > \alpha_2^{\textnormal{e-LOND}}$. Therefore, whenever
    \begin{align}
        \prob\left(E_1 \geq (\delta\gamma_1)^{-1}\text{ and } (\alpha_2^{\textnormal{d-eLOND}})^{-1} \leq E_2 < (\alpha_2^{\textnormal{e-LOND}})^{-1}\right) > 0,
    \end{align}
    donation e-LOND rejects $H_2$ while e-LOND does not, which proves our desired result.
\end{proof}

\subsection{Proof of Theorem~\ref{thm: closed rlond}}\label{subsec:proof-closed-rlond}
\begin{proof}
    The above choice of $E_S$ in \eqref{eq: closed rlond ecol} is a valid e-value for $H_S$ since $f_t$ is a valid calibrator for all $t \in \naturals$ and we are taking a weighted mean of $(f_{|S \cap [i]|}(P_i))_{i \in S}$ with weights $(\gamma_{|S \cap [i]|})_{i \in S}$.
    Thus, since we want to have $R_t \in \Ccal_t$, we need to satisfy for all $S \subseteq [t - 1]$ that
    \begin{align}
        \delta E_{S \cup \{t\}} &= \frac{\ind\{P_t \leq \delta \gamma_{|S| + 1} (|S| + 1) / \ell_{|S| + 1}\}}{\left\lceil (P_t \ell_{|S| + 1} / (\delta \gamma_{|S| + 1})) \vee 1 \right\rceil} + \delta E_S\\
        &\geq
        \FDP_{S \cup \{t\}}(R_{t - 1} \cup \{t\}).
    \end{align}
    If we rearrange this, we get the following condition for all $S \subseteq [t - 1]$:
    \begin{align}
        \frac{\ind\{P_t \leq \delta \gamma_{|S| + 1} (|S| + 1) / \ell_{|S| + 1}\}}{\left\lceil (P_t \ell_{|S| + 1} / (\delta \gamma_{|S| + 1})) \vee 1 \right\rceil}
        \geq \Delta E_S(t).
    \end{align}
    If $\Delta E_S(t) \leq 0$, this constraint is trivially satisfied by any $P_t \in [0, 1]$. Otherwise, we can derive the following condition on $P_t$:
    \begin{align}
        P_t
        \leq
        \frac{\delta \gamma_{|S| + 1}}{\ell_{|S| + 1}}
        \left\lfloor
        \left(\delta \Delta E_S(t)\right)^{-1}\wedge(|S| + 1)
        \right\rfloor.
    \end{align}
    Thus, we get that the above is equivalent to $P_t \leq \alpha_t$ where $\alpha_t$ is given in \eqref{eq:closed-r-lond}. As a result, we get that SupFDR control is ensured by $\overline{\rLOND}$ as a result of \Cref{fact: online supfdr e-closure}.

    To show that $\overline{\rLOND}$ strictly improves over r-LOND, we first prove that the test levels of $\overline{\rLOND}$ are always at least as large as those of r-LOND.
    Because $R_{t - 1} \in \Ccal_{t - 1}$, we have $\FDP_S(R_{t - 1}) \leq \delta E_S $ for all $S \subseteq [t - 1]$. Thus, we also get that $|S \cap R_{t - 1}| \leq \delta E_S |R_{t - 1}|$. As a result
    \begin{align}
        1 + |S \cap R_{t - 1}| - \delta E_S (|R_{t - 1}| + 1) \leq 1 - \delta E_S \leq 1.
    \end{align}
    Equivalently, for any nonvacuous $S$ with $k=|S|+1$ and $r=|R_{t-1}|+1$, we have $(\delta\Delta E_S(t))^{-1}\geq r$. The candidate threshold in \eqref{eq:closed-r-lond} is therefore at least
    \[
        \frac{\delta\gamma_k}{\ell_k}(k\wedge r).
    \]
    Since $k\leq t$ and $i\gamma_i/\ell_i$ is nonincreasing, both $\gamma_i/\ell_i$ and $i\gamma_i/\ell_i$ are nonincreasing enough to imply
    \[
        \frac{\gamma_k}{\ell_k}(k\wedge r)
        \geq
        \frac{\gamma_t}{\ell_t}(t\wedge r).
    \]
    This is exactly the r-LOND level in \eqref{eq: rlond alpha} for the harmonic reshaping function.

    A concise sufficient condition to attain a strict improvement at time $t$ is that the terms being minimized over are all vacuous, i.e., that $\Delta E_S(t) \leq 0$ for all nonempty $S \subseteq [t - 1]$.
    As a result, we get that the minimization in \eqref{eq:closed-r-lond} is achieved by $S = \emptyset$ which in turn implies that $\alpha_t = \delta \gamma_1$. This can be strictly larger than the r-LOND level $\delta\gamma_t((|R_{t-1}|+1)\wedge t)/\ell_t$ under the stated monotonicity condition.

    A concrete instance of this for $t = 2$ is on the event $P_1 \leq \delta\gamma_1$. Then both $\overline{\rLOND}$ and \rLOND{} reject $H_1$ making $R_1 = \{1\}$. Then, we have that
    \begin{align}
        \Delta E_{\{1\}}(2)=\frac{1 + |\{1\}\cap R_1|}{\delta(|R_1| + 1)} - E_{\{1\}}
        = \frac{1}{\delta} - \frac{1}{\delta}=0,
    \end{align}
    because $f_1(P_1)=1/(\delta\gamma_1)$ and hence $E_{\{1\}}=\gamma_1 f_1(P_1)=1/\delta$. Now, we use $\alpha_t^{\rLOND}$ and $\alpha_t^{\overline{\rLOND}}$ to differentiate the test levels of the different procedures. We have that $\alpha_2^{\overline{\rLOND}}=\delta\gamma_1$, while $\alpha_2^{\rLOND}=2\delta\gamma_2/\ell_2$. If $i\gamma_i / \ell_i$ is strictly decreasing at $i=2$, then $\gamma_1 > 2\gamma_2 / \ell_2$, so $\alpha_2^{\overline{\rLOND}} > \alpha_2^{\rLOND}$. Consequently, any distribution that satisfies
    \begin{align}
        \prob\left(P_1 \leq \delta\gamma_1,\ \alpha_2^{\rLOND} < P_2 \leq \alpha_2^{\overline{\rLOND}}\right) > 0,
    \end{align}
    will have $\overline{\rLOND}$ reject $H_2$ while $\rLOND$ does not with positive probability, which establishes strict improvement.
\end{proof}

\subsection{Proof of Theorem~\ref{thm: donation rlond}}\label{subsec:proof-donation-rlond}
\begin{proof}
    Let $E_t = f_t(P_t)$ for each $t \in \naturals$. The calibrator $f_t$ of \citet{vovk_e-values_calibration_2021} guarantees that $\mathbf{E}$ are valid e-values even under arbitrary dependence. Let $\alpha_t^{\textnormal{d-eLOND}}$ denote the test level of donation e-LOND from \eqref{eq:donation-alpha}, and let $\alpha_t^{\textnormal{d-rLOND}}$ denote the test level of donation r-LOND from \eqref{eq:donation-r-lond}. Donation r-LOND therefore enjoys SupFDR control by \Cref{thm:donation-elond}.

    To show that rejecting when $P_t \leq \alpha_t^{\textnormal{d-rLOND}}$ is equivalent to rejecting when $f_t(P_t) \geq (\alpha_t^{\textnormal{d-eLOND}})^{-1}$, we directly compare rejection decisions. By the definition of $f_t$, the donation e-LOND condition $f_t(P_t) \geq (\alpha_t^{\textnormal{d-eLOND}})^{-1}$ is equivalent to
    \begin{align}
        P_t \leq \frac{\delta\gamma_t t}{\ell_t}
        \quad\text{and}\quad
        \left\lceil\left(P_t \ell_t / (\delta \gamma_t)\right) \vee 1 \right\rceil
        \leq
        \frac{\alpha_t^{\textnormal{d-eLOND}}}{\delta \gamma_t}.
    \end{align}
    Since the left-hand side is integer valued, the second inequality is equivalent to
    \begin{align}
        \left\lceil\left(P_t \ell_t / (\delta \gamma_t)\right) \vee 1 \right\rceil
        \leq
        \left\lfloor \frac{\alpha_t^{\textnormal{d-eLOND}}}{\delta \gamma_t} \right\rfloor.
    \end{align}
    Combining with the indicator constraint $P_t \leq \delta\gamma_t t/\ell_t$, we get the equivalent condition
    \begin{align}
        P_t
        \leq
        \frac{\delta \gamma_t}{\ell_t}
        \left(
            \left\lfloor \frac{\alpha_t^{\textnormal{d-eLOND}}}{\delta \gamma_t} \right\rfloor
            \wedge t
        \right).
    \end{align}
    Substituting \eqref{eq:donation-alpha} for $\alpha_t^{\textnormal{d-eLOND}}$ yields exactly the threshold in \eqref{eq:donation-r-lond}. Thus the rejection sets of donation r-LOND coincide with those of donation e-LOND, so the SupFDR control conclusion transfers.

    Lastly, we show that donation r-LOND can strictly improve over r-LOND. Let $\alpha_t^{\textnormal{d-rLOND}}$ and $\alpha_t^{\textnormal{r-LOND}}$ denote the test levels of donation r-LOND and r-LOND respectively.
    A concise sufficient condition for strict inequality at time $t$ is that both procedures agree up to time $t - 1$, with $r_t\coloneqq |R_{t - 1}| + 1 < t$, and the donation wealth is large enough to increase the integer factor:
    \[
        \left\lfloor
                \frac{r_t}{1 - (\delta r_t\bar{W}_t \wedge 1)}
        \right\rfloor > r_t .
    \]
    In general, using $r_t \leq t$ together with $0 \leq \delta r_t\bar{W}_t \wedge 1 \leq 1$, we have
    \begin{align}
        \alpha_t^{\textnormal{d-rLOND}}
        &=
        \frac{\delta\gamma_t}{\ell_t}
        \left(
            \left\lfloor    
                \frac{|R_{t - 1}| + 1}{1 - (\delta (|R_{t - 1}| + 1)\bar{W}_t \wedge 1)}
            \right\rfloor
            \wedge t
        \right)
        \geq
        \frac{\delta\gamma_t}{\ell_t}\big((|R_{t - 1}| + 1)\wedge t\big)\\
        &=
        \alpha_t^{\textnormal{r-LOND}}.
    \end{align}

    Strict inequality holds whenever the displayed floor is larger than $(|R_{t-1}|+1)\wedge t$. Thus, whenever for some $t$,
    \begin{align}
        \prob\left(\alpha_t^{\textnormal{r-LOND}} < P_t \leq \alpha_t^{\textnormal{d-rLOND}}\right) > 0,
    \end{align}
    donation r-LOND rejects $H_t$ while r-LOND does not, which gives strict improvement. Such events are easy to construct by placing positive mass on earlier p-values that create positive donation wealth while leaving $r_t<t$.
\end{proof}

\section{Simulation extensions}\label{sec:simulation-extensions}
We provide additional simulation results and construction details referenced in the main text.

\subsection{Additional simulations}\label{sec: additional simulations}

We also include simulations in the same setup as \Cref{sec: simulations}, although with $\mu=4$ in \Cref{fig:local-dep-simulation-mu4}. We see similar results, with the closed and donation variants of each procedure improving over the corresponding baselines as the non-null fraction $\pi_1$ increases.
\begin{figure*}[th]
    \centering
    \begin{subfigure}{.5\textwidth}
        \centering
\includegraphics[width=\textwidth,trim={0 0 0 0},clip]{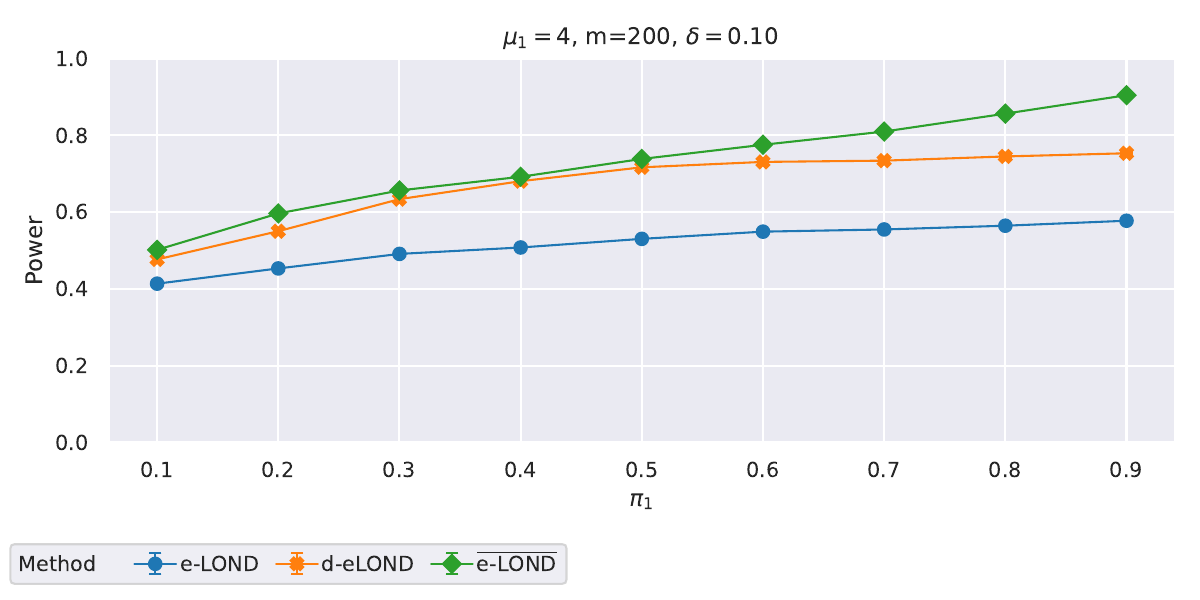}\caption{E-value procedures.}
    \end{subfigure}\begin{subfigure}{.5\textwidth}
        \centering
\includegraphics[width=\textwidth,trim={0 0 0 0},clip]{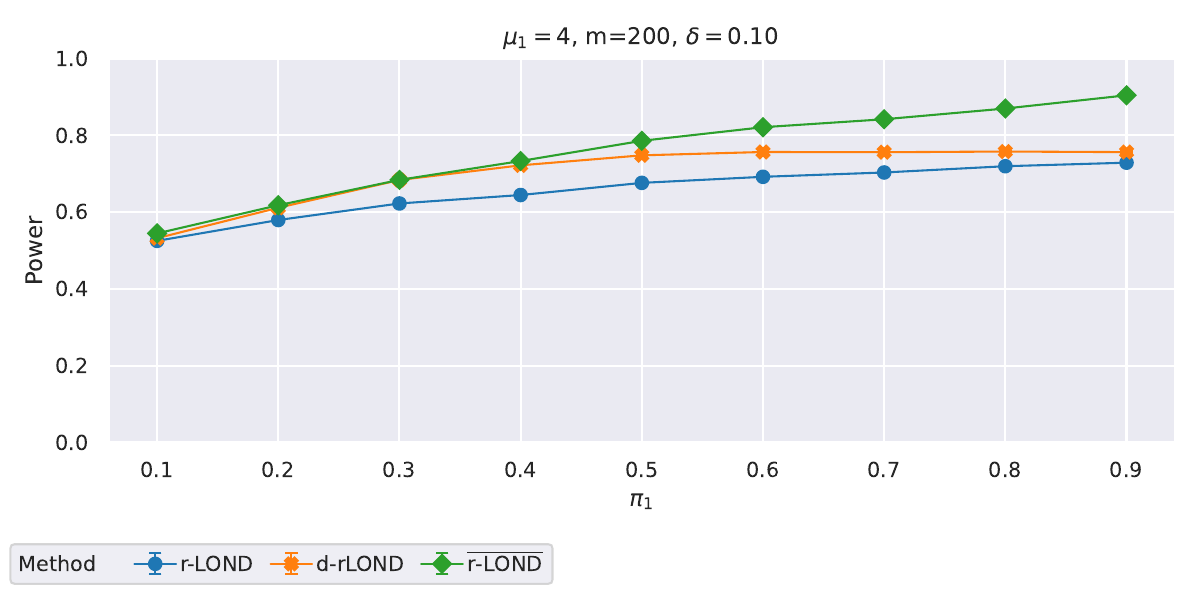}\caption{P-value procedures.}
    \end{subfigure}
    \caption{Power-only local dependence simulation results over 200 trials for alternative signals of $\mu = 4$ and $\delta = 0.1$. In both e-value and p-value families, donation and closed variants improve over the corresponding baselines as the non-null fraction $\pi_1$ increases.}
    \label{fig:local-dep-simulation-mu4}
\end{figure*}

\subsection{\revheading{Real data e-value formulation}}\label{subsec:real-data-annotation-metrics}
\begin{revblock}
The real-data experiment in \Cref{tab:real-data-summary} uses the NYC taxi dataset with annotated anomaly windows. Let $Z_t$ denote the STL residual series, let $C=\{1,\ldots,n_{\mathrm{cal}}\}$ be the calibration indices, and let $W_t=\{\max(1,t-n_{\mathrm{win}}),\ldots,t-1\}$ be the rolling alternative window. We use two density estimates. First, $f_0$ is the baseline residual density: a one-dimensional kernel density estimate fit to the calibration residuals $\{Z_i:i\in C\}$. Second, $f_t$ is a time-local alternative residual density for the $t$th test: we fit a bivariate kernel density estimate to the recent time-residual pairs $\{(i,Z_i):i\in W_t\}$, then evaluate this fitted density with the time coordinate fixed at $t$. Thus $f_t(z)$ is the local density assigned to residual value $z$ near time $t$, while $f_0(z)$ is the calibration baseline density assigned to the same residual value. The likelihood ratio e-value used by the procedures is
\[
    E_t^{\mathrm{taxi}}
    \coloneqq
    \frac{f_t(Z_t)/f_0(Z_t)}
    {|C|^{-1}\sum_{i\in C} f_t(Z_i)/f_0(Z_i)}.
\]
This normalized likelihood-ratio score is used as an empirical anomaly e-value proxy in the reproduction. A formal e-value guarantee for such fitted density-ratio scores requires additional sample-splitting or exchangeability assumptions; for example, with an independently trained $f_t/f_0$, the conformalized score
\[
    E_t^{\mathrm{conf}}
    \coloneqq
    \frac{(|C|+1) f_t(Z_t)/f_0(Z_t)}
    {f_t(Z_t)/f_0(Z_t)+\sum_{i\in C} f_t(Z_i)/f_0(Z_i)}
\]
has the usual calibration interpretation under exchangeability of the test residual and calibration residuals. We therefore interpret \Cref{tab:real-data-summary} as an empirical anomaly-detection illustration rather than as a separate verification of the formal SupFDR theorem under estimated density ratios.
\end{revblock}

\subsection{Simulation details}\label{sec:local-dep-details}
The main simulations use the Gaussian local-dependence model described in \Cref{sec: simulations}. For each hypothesis $t$, draw a latent indicator $A_t\sim\mathrm{Bernoulli}(\pi_1)$, set $\mu_t=\mu A_t$, and generate a Gaussian vector $(X_1,\ldots,X_m)$ with mean vector $(\mu_1,\ldots,\mu_m)$ and covariance matrix
\[
    \Sigma_{ij}=0.5^{|i-j|}\ind\{|i-j|\leq L\}.
\]
We use $L=100$ and verify numerically that the resulting finite covariance matrix is positive semidefinite for the simulated values of $m$. Under the null, each marginal observation is $N(0,1)$; under the alternative, each marginal observation is $N(\mu,1)$. The e-value for testing $H_t:\mu_t=0$ against the simple alternative $\mu_t=\mu$ is the Gaussian likelihood ratio
\[
    E_t=\exp\{\mu X_t-\mu^2/2\}.
\]
The p-value counterpart is the one-sided Gaussian p-value
\[
    P_t=1-\Phi(X_t).
\]
Power is reported as the realized fraction of nonnull hypotheses rejected. In simulations where we report empirical error, the SupFDR estimate is computed as the trial average of $\sup_{t\leq m}\FDP_{\Ncal}(R_t)$; some supplemental plots also include final-time FDR diagnostics for comparison with earlier online-testing literature.

\section{Improvements via donation beyond online multiple testing}\label{sec: arc donation}

We will discuss several improvements we can make to e-value procedures using the donation framework that is beyond the online multiple testing setting. In particular, we will consider how we can use the donation framework to improve methods for the acceptance-to-rejection (ARC) model of \citet{fischer_online_generalization_2025}, the decision deadlines setup of \citet{fisher_online_control_2022}, as well as the traditional offline multiple testing setting with a finite number of hypotheses.

\subsection{Online donation e-BH for acceptance-to-rejection (ARC)}
The acceptance-to-rejection (ARC) setup of \citet{fischer_online_generalization_2025} assumes that once one makes a rejection on a hypothesis, they cannot revoke it at a future time step, but allows one to make rejections on all previous unrejected hypotheses. This is equivalent to restricting the rejection set at each time step to be nested, i.e., $R_1 \subseteq R_2, \dots$. 

\paragraph{The online e-BH procedure} In the ARC setting, \citet{fischer_online_generalization_2025} showed that one can apply a weighted version of the e-BH procedure to an infinite stream of hypotheses using the weighted self-consistency framework and maintain FDR control. The concept of weighted self-consistency is a generalization of the online \eBH{} and \ELOND{} procedure, since the discovery sets of both are included in the simultaneous weighted self-consistency collection of discovery sets in \eqref{eq:online weighted self-consistency}, which guarantees control over the supremum of FDP of all discovery sets in the collection.

For a given fixed sequence $\boldsymbol{\gamma}$, the \emph{online e-BH procedure} makes the following number of discoveries at time $t$:
\begin{align}
    r_t^{\textnormal{o-eBH}} \coloneqq \max\left\{r \in [t]: \sum_{i \in [t]} \ind\left\{E_i \geq \frac{1}{\delta \gamma_i r}\right\} \geq r\right\}, \label{eq: online-ebh}
\end{align}
with $r_t^{\textnormal{o-eBH}} = 0$ if the set is empty. The rejection set is then defined as:
\begin{align}
   R_t = \left\{i \in [t]: E_i \geq \frac{1}{\delta \gamma_i r_t^{\textnormal{o-eBH}}}\right\}. 
\end{align}
 
Using the donation framework, we can improve online e-BH. First, we define the notion of weighted order e-values, i.e., let $\gamma_{(i):t}$ and $E_{(i):t}$ be the values of $\gamma_j$ and $E_j$ corresponding to the $i$th largest $\gamma_j E_j$ among $j \in [t]$. Thus, we get the \emph{online donation e-BH procedure} as follows. We can define the number of discoveries made as
\begin{align}
    r_t^{\textnormal{o-DeBH}}
    \coloneqq
    \max\Bigg\{r \in [t]:
    &\sum_{i \in [r]} (\gamma_{(i):t} E_{(i):t} - (\delta r)^{-1}) \wedge \gamma_{(i):t} \notag\\
    &\quad + \sum_{i \in \{r + 1, \dots, t\}} \gamma_{(i):t} (E_{(i):t} \wedge 1) \geq 0\Bigg\}, \label{eq: online donation ebh}
\end{align}
As a result, the discovery set $R_t$ simply rejects the $r_t^{\textnormal{o-DeBH}}$ largest indices of $\gamma_i E_i$ among $i \in [t]$. 
\begin{theorem}[Online donation e-BH controls SupFDR]\label{thm:donation-online-ebh}
    Online donation e-BH with the aforementioned rejection sets $\mathbf{R}$ satisfies $\SupFDR(\mathbf{R}) \leq \delta$ for arbitrarily dependent e-values, and strictly improves over online e-BH.
\end{theorem}

\subsubsection{Simulation results}
We compare online e-BH and donation online e-BH in the ARC setting using the
same local dependence simulation setup as in the main simulations section.
\begin{figure}[t]
    \centering
    \includegraphics[width=.6\textwidth]{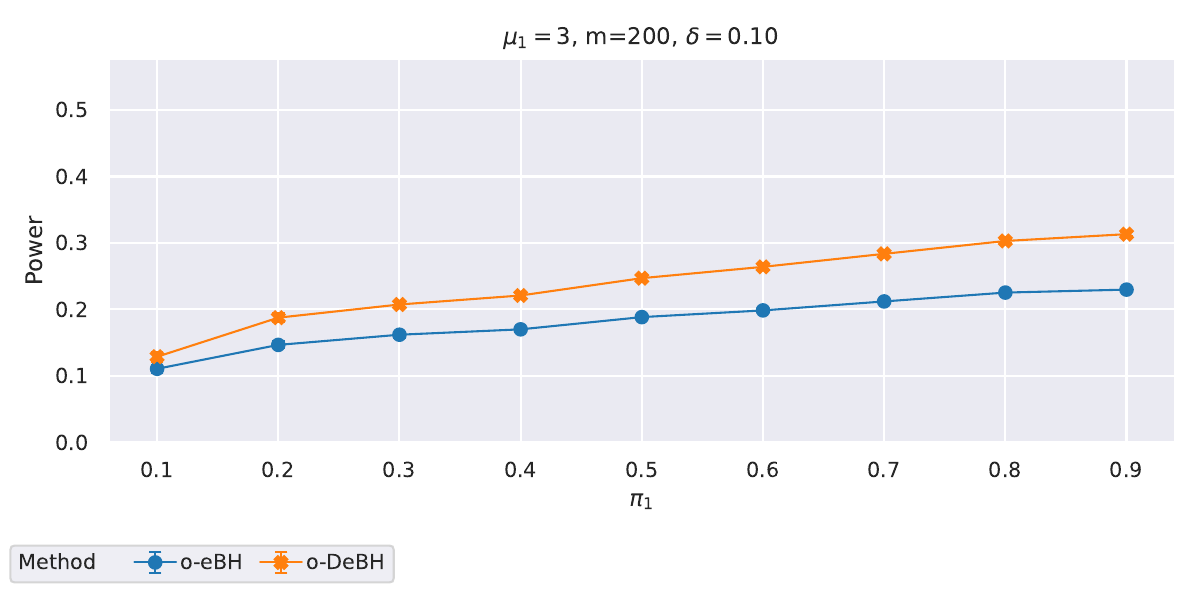}

    \includegraphics[width=.6\textwidth]{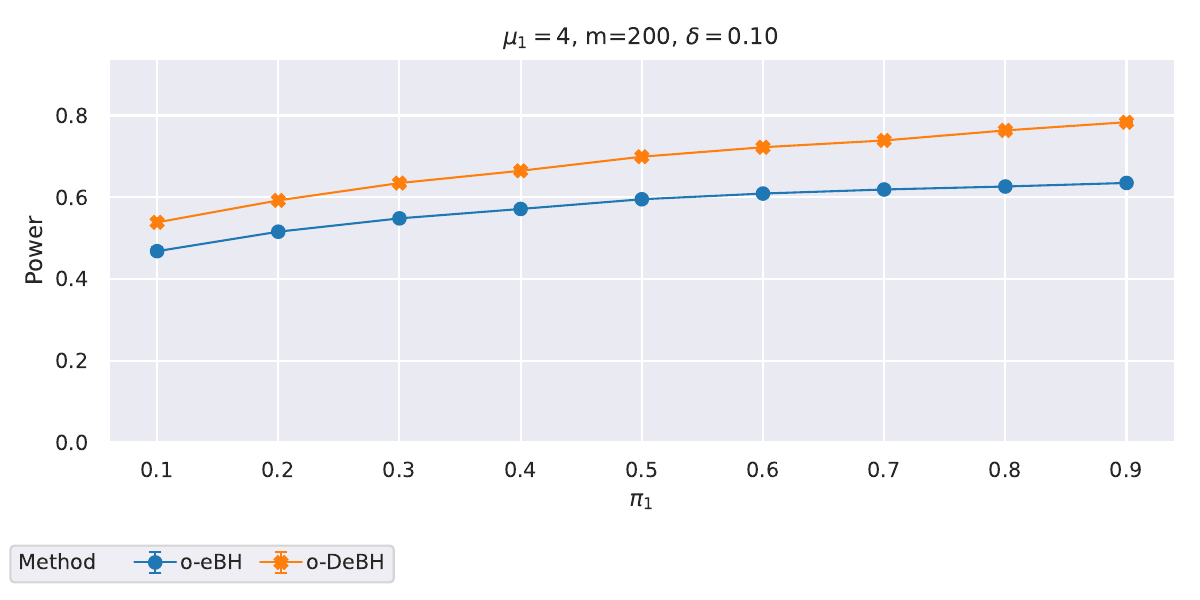}
    \caption{Power for online e-BH and donation online e-BH in ARC under the same setup as the main simulations section.}
    \label{fig:supp-arc-ebh-sim-power}
\end{figure}
\begin{figure}[t]
    \centering
    \includegraphics[width=.6\textwidth]{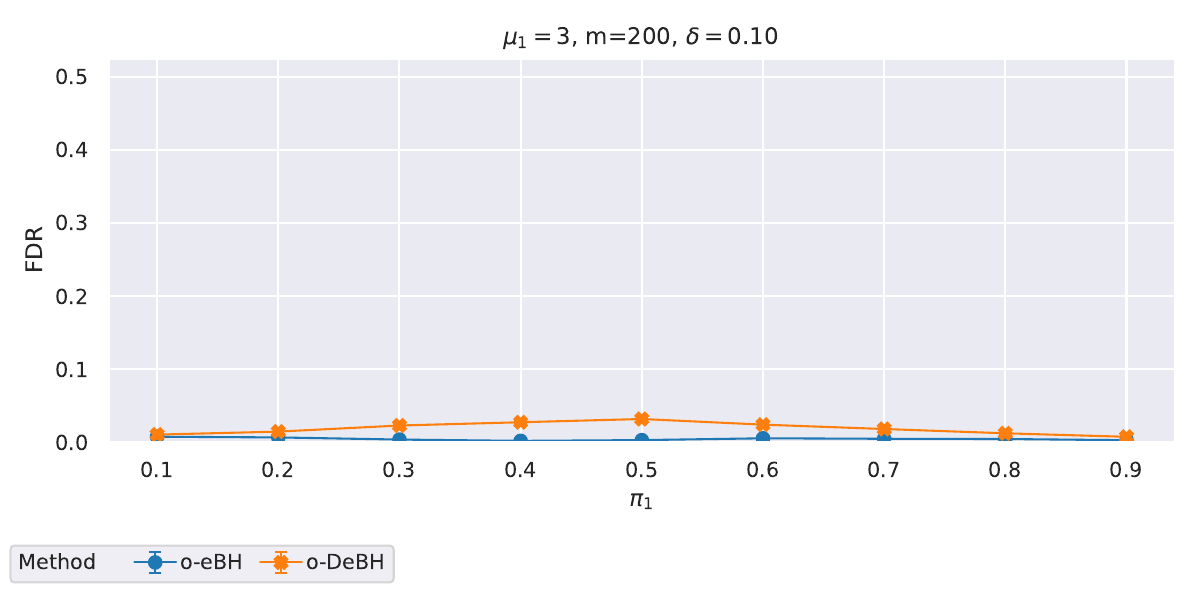}

    \includegraphics[width=.6\textwidth]{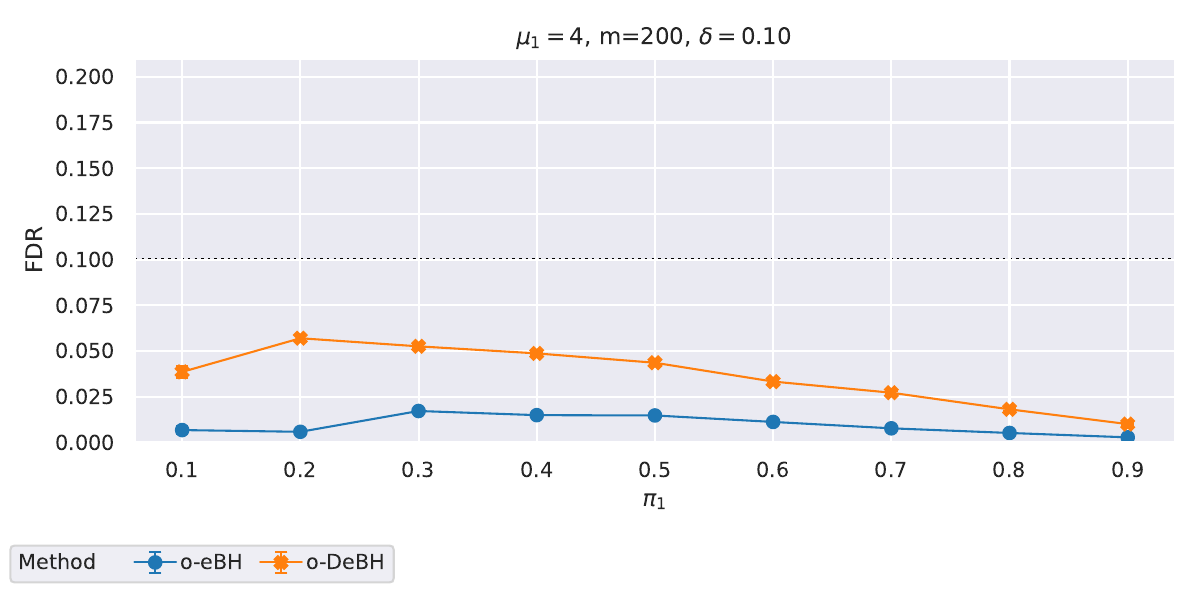}
    \caption{Empirical error diagnostics for online e-BH and donation online e-BH in ARC. Both methods stay controlled at the target level $\delta = 0.1$.}
    \label{fig:supp-arc-ebh-sim-fdr}
\end{figure}

\subsection{Donation e-TOAD for decision deadlines}
On the other hand, \citet{fisher_online_control_2022} studies an intermediate regime where each hypothesis $t$ has a deterministic deadline $d_t \geq t$, i.e., one must make a rejection decision at time $d_t$ for the $t$th hypothesis, or else the null hypothesis will be accepted and remain unrejected permanently. Let $\Acal_t \coloneqq \{i \in [t]: d_i \leq t\}$ denote the set of hypotheses whose deadlines have not yet passed by time $t$. The ARC model corresponds to setting $d_t = \infty$ for each $t \in \naturals$, while the classical online multiple testing setting corresponds to making $d_t = t$ for each $t \in \naturals$.

\paragraph{The e-TOAD procedure} Fisher's TOAD rule \citep{fisher_online_control_2022} now restricts rejection sets at certain deadlines. Let $m_t \coloneqq |\Acal_t|$ be the number of hypotheses whose deadlines have arrived by time $t$. Now, we let $\gamma_{(i):\Acal_t}$ and $E_{(i):\Acal_t}$ denote the values of $\gamma_i$ and $E_i$ corresponding to the $i$th largest $\gamma_i E_i$ among $i \in \Acal_t$. Define
\begin{align}
    r_t^{\textnormal{eTOAD}} \coloneqq \max\left\{r \in \{|R_{t - 1} \setminus \Acal_t|, \dots, m_t\}: \sum_{i \in \Acal_t} \ind\left\{E_i \geq \frac{1}{\delta \gamma_i r}\right\} \geq r - |R_{t - 1} \setminus \Acal_t|\right\},
\end{align} where $r_t^{\textnormal{eTOAD}} = 0$ if no such $r$ exists. Then, we define the corresponding discovery set as 
\begin{align}
    R_t = R_{t - 1} \cup \left\{i \in \Acal_t: E_i \geq \frac{1}{\delta \gamma_i r_t^{\textnormal{eTOAD}}}\right\}.
\end{align}
When $d_t = t$ for all $t \in \naturals$, this is equivalent to e-LOND, and if $d_t = \infty$, then this is the same as online e-BH. We can then define \emph{online donation e-TOAD} using the following quantities:
\begin{align}
    \bar{W}_t^{\textnormal{DeTOAD}}(r)
    &\coloneqq
    \sum_{i \in  R_{t - 1} \setminus \Acal_t} (\gamma_i E_i - (\delta r)^{-1}) \wedge \gamma_i \notag\\
    &\quad + \sum_{i \in [t] \setminus (\Acal_t \cup R_{t - 1})} \gamma_i (E_i \wedge 1),\\
     r_t^{\textnormal{DeTOAD}}
     &\coloneqq
     \max\Bigg\{r \in \{|R_{t - 1} \setminus \Acal_t|, \dots, m_t\}: \notag\\
     &\quad \sum_{i \in [r - |R_{t - 1} \setminus \Acal_t|]} (\gamma_{(i):\Acal_t} E_{(i):\Acal_t} - (\delta r)^{-1}) \wedge \gamma_{(i):\Acal_t} \notag\\
     &\quad + \sum_{i \in \{r - |R_{t - 1} \setminus \Acal_t| + 1, \dots, m_t\}} \gamma_{(i):\Acal_t} (E_{(i):\Acal_t} \wedge 1) \notag\\
     &\quad + \bar{W}_t^{\textnormal{DeTOAD}}(r) \geq 0\Bigg\},\label{eq: donation-etoad}
\end{align}
with $r_t^{\textnormal{DeTOAD}} = 0$ if no such $r$ exists. Here, $\bar{W}_t^{\textnormal{DeTOAD}}(r)$ is the excessive ``wealth'' that can be donated to (or accounted for from) e-values in $\Acal_t$ if a total of $r$ discoveries are made at time $t$. The rejection set $R_t$ is then defined as making $r_t^{\textnormal{DeTOAD}} - |R_{t - 1}|$ new discoveries corresponding to the indices in $\Acal_t$ with the largest $\gamma_i E_i$ values.

\begin{theorem}[Donation e-TOAD controls SupFDR]\label{thm:donation-etoad}
    Donation e-TOAD with the aforementioned discovery sets satisfies $\SupFDR(\mathbf{R}) \leq \delta$ for arbitrarily dependent e-values, and strictly improves over e-TOAD.
\end{theorem}

\begin{proof}[Proof of \Cref{thm:donation-online-ebh,thm:donation-etoad}]
    The proofs of both of these theorems are similar to that of donation e-LOND: at each time $t$, there exists a $\boldsymbol{\gamma}$-weighted donation sequence $\mathbf{B}^{(t)}$ with $B_i^{(t)} \geq -(E_i \wedge 1)$ and $\sum_{i \in [t]} \gamma_i B_i \leq 0$ such that the compound e-values $\tilde{E}_i^{(t)} = E_i + B_i$ satisfy $\gamma_i \tilde{E}_i^{(t)} \geq (\delta |R_t|)^{-1}$ for all $i \in R_t$. By \Cref{prop: donation compound}, $(\tilde{E}_i^{(t)})_{i \in [t]}$ are valid $\boldsymbol{\gamma}$-weighted compound e-values. The weighted self-consistency collection $\Ccal(\mathbf{B}^{(t)})$ from \eqref{eq:online weighted self-consistency} applied to these compound e-values contains the rejection set $R_t$ for both procedures. Hence, we get SupFDR control via \Cref{prop: sup donation}.

    For online donation e-BH, set $r_t^\star=r_t^{\textnormal{o-DeBH}}$; for donation e-TOAD, set $r_t^\star=r_t^{\textnormal{DeTOAD}}$. The choice of $B_i^{(t)}$ for both procedures can be chosen as follows:
    \begin{align}
        B_i^{(t)} \coloneqq \begin{cases}
            ((\delta \gamma_i r_t^\star)^{-1} - E_i) \vee -1, & i \in R_t,\\
            -(E_i \wedge 1), & i \not\in R_t.
        \end{cases}
    \end{align}
    In the e-TOAD case, already rejected hypotheses in $R_{t-1}\setminus\Acal_t$ are included in $R_t$ and therefore use the first branch, while unrejected inactive hypotheses use the second branch. By definition of each procedure this $\mathbf{B}^{(t)}$ is a $\boldsymbol{\gamma}$-weighted donation sequence in both cases. Thus, we have shown our desired results.
\end{proof}

Similar to donation e-LOND, both donation online e-BH and donation e-TOAD strictly improve their non-donation counterparts since they consider a superset of rejection sets at each time $t$. Set-inclusion comparisons assume a fixed deterministic tie-breaking convention for hypotheses tied in $\gamma_iE_i$; alternatively, one may reject all hypotheses tied at the selected threshold in both the baseline and donation procedures.

\subsection{Donation e-BH for offline multiple testing}\label{sec: offline donation}

We treat the offline batch as the $t = m$ snapshot of the ARC model with all deadlines at $\infty$ and with $\gamma_1 = \dots = \gamma_m = 1/m$ and $\gamma_{t} = 0$ for all $t > m$. Thus, we can order hypotheses directly by e-values $E_i$, writing $E_{(1)} \geq \dots \geq E_{(m)}$. In this case, we can view eBH and donation eBH as taking $R_m$ of online eBH and online donation eBH, respectively.

\paragraph{Baseline e-BH.} The classical e-BH rule has the same weighted self-consistency form as \eqref{eq: online-ebh}:
\begin{align}
    r^{\textnormal{eBH}}
    &\coloneqq \max\left\{r \in [m]: \sum_{i \in [m]} \ind\!\left\{E_i \geq \frac{m}{\delta r}\right\} \geq r\right\}, \notag\\
    R^{\textnormal{eBH}}
    &= \left\{i \in [m]: E_i \geq \frac{m}{\delta r^{\textnormal{eBH}}}\right\}.
\end{align}

\paragraph{Donation-derived compound e-values.} Here, we only need to consider vectors $\mathbf{B} = (B_i)_{i \in [m]}$. A vector is a valid donation sequence if $B_i \geq -(E_i \wedge 1)$ and $\sum_{i = 1}^m B_i \leq 0$. Let $\hat{E}_i \coloneqq E_i + B_i$.
\begin{proposition}\label{prop:donation-compound-evalues}
    If $(E_i)_{i \in [m]}$ are e-values and $\mathbf{B}$ is a valid donation sequence, then $(\hat{E}_i)_{i \in [m]}$ are compound e-values, i.e., $\sum_{i \in \Ncal} \expect[\hat{E}_i] \leq m$.
\end{proposition}
This follows from \Cref{prop: donation compound}.

\paragraph{Donation e-BH.} The offline donation method can be seen as using \eqref{eq: online donation ebh} with $t = m$:
\begin{align}
    r^{\textnormal{DeBH}} \coloneqq \max\Big\{r \in [m]: \sum_{i \in [r]} (\gamma_{(i)} E_{(i)} - (\delta r)^{-1}) \wedge \gamma_{(i)} + \sum_{i = r + 1}^m \gamma_{(i)} (E_{(i)} \wedge 1) \geq 0 \Big\}. \label{eq:donation-ebh-discovery-cond}
\end{align} Then, we let the discovery set $R^{\textnormal{DeBH}}$ reject the $r^{\textnormal{DeBH}}$ largest e-values.
\begin{theorem}
    Donation \eBH{} satisfies $\expect[\FDP_{\Ncal}(R^{\textnormal{DeBH}})] \leq \delta$ for arbitrarily dependent e-values and strictly improves over \eBH.
\end{theorem}
The proof follows from \Cref{thm:donation-online-ebh} since the offline method can be seen as a special case of the online ARC model.

\begin{remark}
    While we have primarily discussed e-value based methods in this section, our results directly imply improvements of p-value methods. This includes the online BH method of \citet{fischer_online_generalization_2025} and the offline Su method \citep{su_fdr-linking_2018} using the calibrator developed in \citet{xu_bringing_closure_2025a}, as well as the Benjamini-Yekutieli \citep{benjamini_control_false_2001} methods via the calibrator specified in \eqref{eq: by calibrator}. We will not go into details here; the improvements follow directly from calibrating p-values to e-values and then applying one of the online or offline donation e-BH procedures, or the e-TOAD procedure if one is to use p-values in the decision deadlines setting.
\end{remark}

\section{Randomization for donation algorithms}\label{sec: randomized donation}
\citet{xu_more_powerful_2023} introduced the notion of using randomization in the form of stochastic rounding to improve the power of a variety of multiple testing procedures, and \citet{xu_online_multiple_2023} was able to show that randomized versions of e-LOND and r-LOND can be derived using this technique. We extend the use of randomization to donation procedures, and show that one can utilize randomization to further improve the power of donation e-LOND and donation r-LOND. The key idea we recognize here is that while donation procedures utilize part of the excess wealth of e-values over the rejection threshold, it does not use it completely. Thus, when there is excess wealth remaining even after donation (e.g., an e-value is just slightly below the threshold of rejection even after donating), we can use it to improve the power of a procedure.

A stochastically rounded e-value (or compound e-value) is one where we have an e-value (or any nonnegative random variable) $X$ and a test level $\hat\alpha \in (0, 1]$ that might be arbitrarily dependent on $X$, and we produce the following random variable
\begin{align}
    S_{\hat\alpha}(X) \coloneqq \ind\left\{U \leq \hat\alpha X\right\}\hat\alpha^{-1}.
\end{align} Here, $U$ is a uniform random variable on $[0, 1]$ that is independent of both $X$ and $\hat\alpha$, i.e., produced through external randomness. It is easy to see that $\expect[S_{\hat\alpha}(X)] \leq \expect[X]$. Thus, we can replace a compound e-value or an e-value with its stochastically rounded version and still maintain the validity properties of interest. 

As a result of the flexibility of stochastic rounding, there can be many ways to incorporate it into the donation framework. We will focus on improving donation e-LOND as an example, and show there is a simple way that will allow it to strictly improve over donation e-LOND.

We first define a restricted version of stochastic rounding, where we only round the part of an unrejected e-value that cannot be utilized by the donation framework, i.e., $(E_t - 1)$ if $E_t \geq 1$ and 0 otherwise. Thus, for $\hat\alpha\in(0,1)$ we define the \emph{restricted stochastic rounding} of $X$ at level $\hat\alpha$ as
\begin{align}
    \bar{S}_{\hat\alpha}(X) \coloneqq \begin{cases}
        X & X \leq 1 \text{ or }X \geq \hat\alpha^{-1},\\
        1 + \ind\left\{U\leq \frac{\hat\alpha (X- 1)}{1 - \hat\alpha}\right\} (\hat\alpha^{-1} - 1), & X \in (1, \hat\alpha^{-1}).
    \end{cases}
\end{align}
If $\hat\alpha\geq1$, the hypothesis is already rejected by any e-value $X\geq1$, and we use the convention $\bar S_{\hat\alpha}(X)=X$ rather than applying the fractional rounding formula.

Thus, we can apply the donation e-LOND sequence of test levels $\boldsymbol{\alpha}$ in \eqref{eq:donation-alpha} to e-values $\bar{S}_{\alpha_1}(E_1), \bar{S}_{\alpha_2}(E_2), \dots$ and refer to this as \emph{randomized donation e-LOND}. We first define the following quantities.
\begin{align}
    \hat{\alpha}_t &\coloneqq \frac{\delta \gamma_t(|R_{t - 1}| + 1)}{1 - (\delta(|R_{t - 1}| + 1)\bar{R}_t \wedge 1)}.\\
    \bar{R}_t &\coloneqq \sum_{i \in R_{t - 1}} (\gamma_i\bar{S}_{\hat{\alpha}_i}(E_i) - \delta(|R_{t - 1}| + 1)^{-1}) \wedge \gamma_i + \sum_{i \in [t - 1] \setminus R_{t - 1}} \gamma_i(\bar{S}_{\hat\alpha_i}(E_i) \wedge 1).\\
    &= \sum_{i \in R_{t - 1}} (\gamma_i\bar{S}_{\hat{\alpha}_i}(E_i) - \delta(|R_{t - 1}| + 1)^{-1}) \wedge \gamma_i + \sum_{i \in [t - 1] \setminus R_{t - 1}} \gamma_i(E_i \wedge 1).
\end{align}
$\hat{\alpha}_t$ is the threshold for the $t$th hypothesis such that it will be deterministically rejected (i.e., doesn't rely on randomization). $\bar{R}_t$ is the analog of $\bar{W}_t$ for donation e-LOND, but utilizes the stochastically rounded e-values instead. We then observe that we reject the $t$th hypothesis when
\begin{align}
    \alpha_t = \frac{\hat\alpha_t}{U_t(1 - \hat\alpha_t) + \hat\alpha_t}.
\end{align}
where $U_1, U_2, \dots$ are uniform random variables on $[0, 1]$ independent of $\mathbf{E}$.
\begin{theorem}
The randomized donation e-LOND algorithm ensures control of SupFDR, and strictly improves over donation e-LOND.
\end{theorem}
\begin{proof}
We note SupFDR control arises from the fact that $\bar{S}_{\hat\alpha_t}(E_1), \dots$ are valid e-values due to the definition of restricted stochastic rounding. Thus, \Cref{thm:donation-elond} immediately implies SupFDR control. We can see the strict improvement via the fact that $U_t$ has nonzero chance of increasing $\alpha_t$ over $\hat\alpha_t$, which also is at least as large as $\alpha_t$ of donation e-LOND defined in \eqref{eq:donation-alpha} via construction. Thus, we have shown our desired result.
\end{proof}

\subsection{Simulation results}
We compare donation e-LOND and randomized donation e-LOND under the same
local dependence simulation setup as in the main simulations section.
\begin{figure}[t]
    \centering
    \includegraphics[width=.6\textwidth]{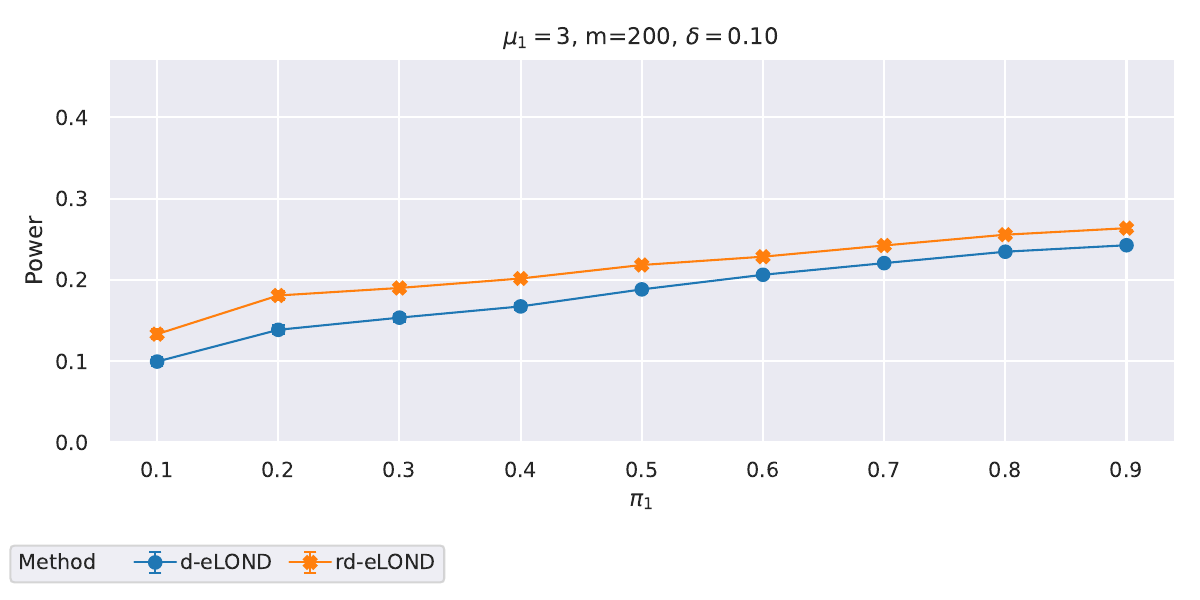}

    \includegraphics[width=.6\textwidth]{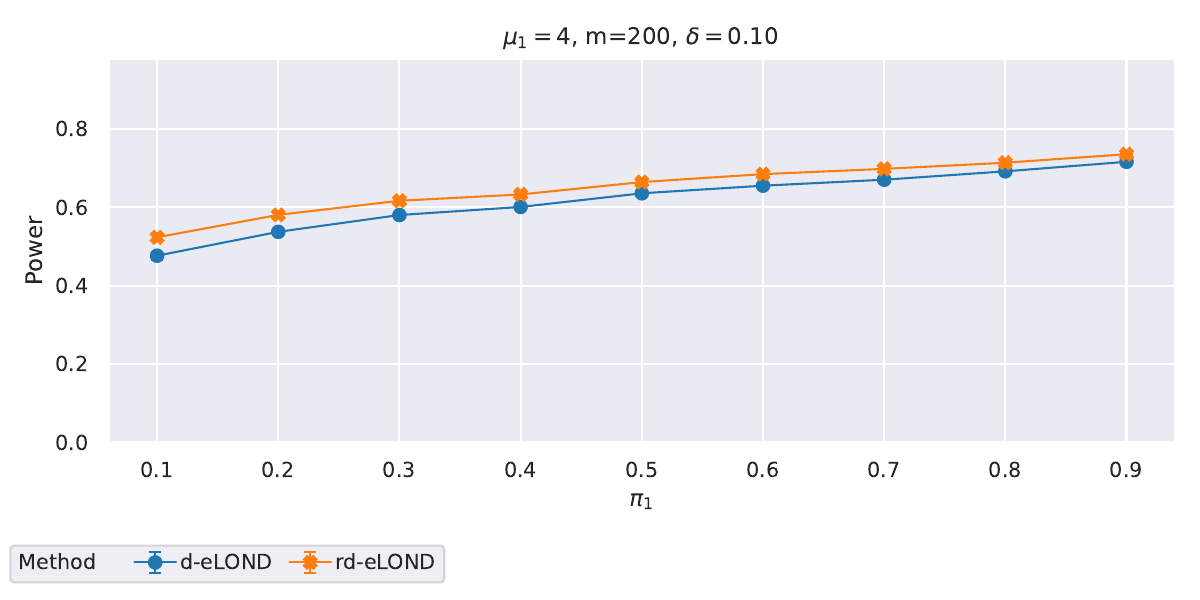}
    \caption{Power for donation e-LOND and randomized donation e-LOND under the same setup as the main simulations section.}
    \label{fig:supp-randomized-donation-sim-power}
\end{figure}
\begin{figure}[t]
    \centering
    \includegraphics[width=.6\textwidth]{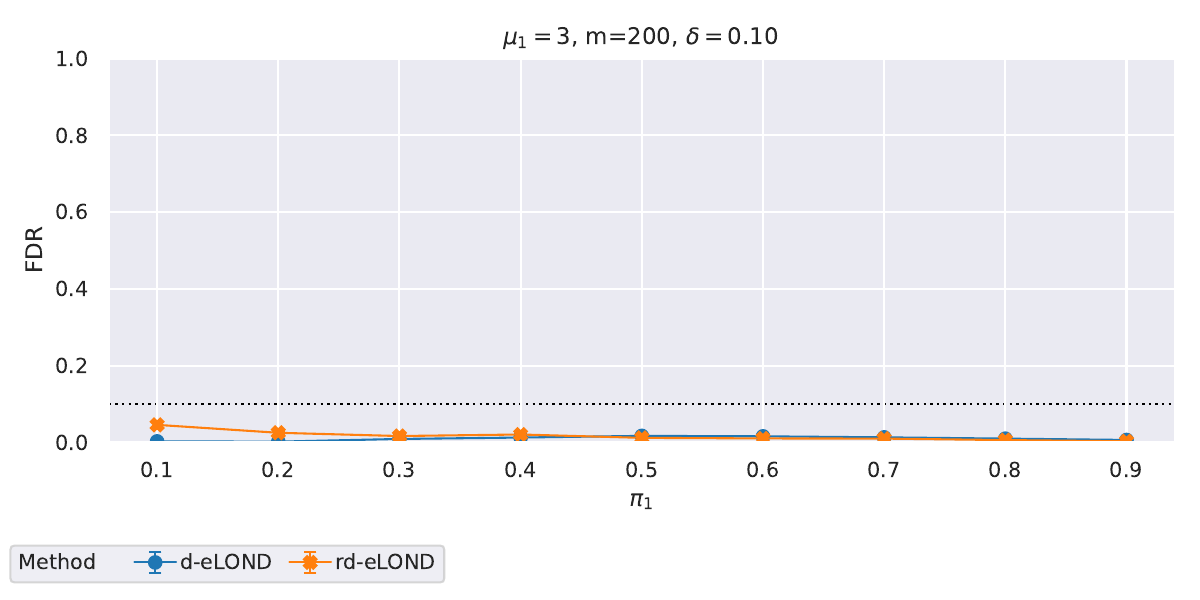}

    \includegraphics[width=.6\textwidth]{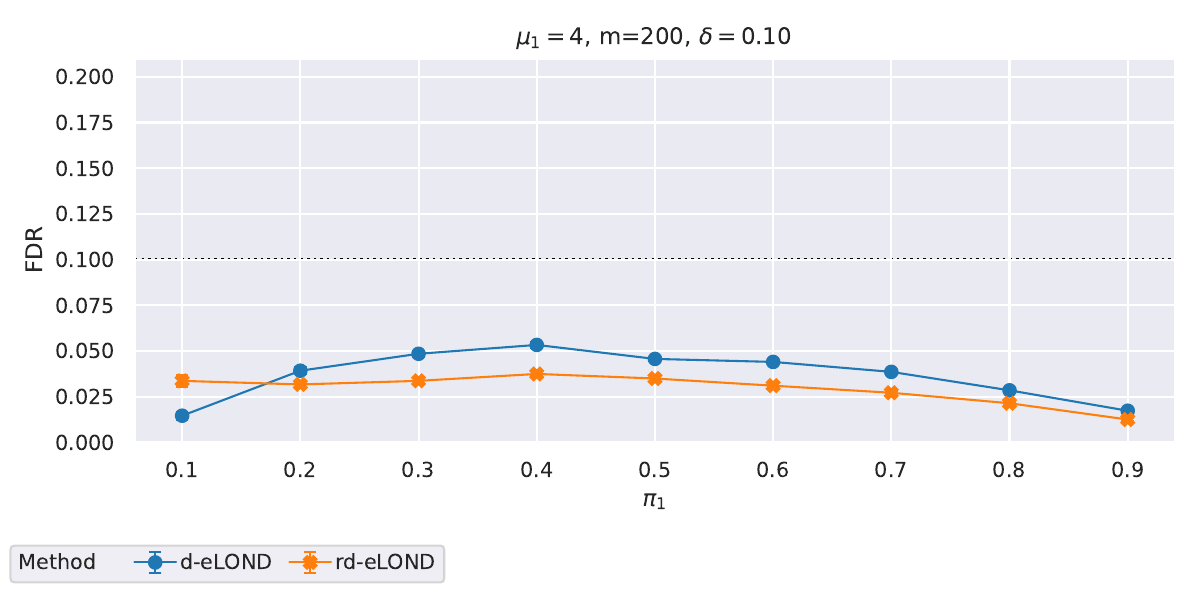}
    \caption{Empirical error diagnostics for donation e-LOND and randomized donation e-LOND. Both methods stay controlled at the target level $\delta = 0.1$ in this simulation.}
    \label{fig:supp-randomized-donation-sim-fdr}
\end{figure}
 
\end{document}